\documentclass[11pt]{article}
\usepackage{verbatim} % USE \BEGIN{COMMENT}..\END{COMMENT}
\usepackage[usenames, dvipsnames]{color}

\usepackage{bbm}
\usepackage{mathtools}
\usepackage{amsmath}
\usepackage{amsthm, thmtools}
\usepackage{amstext,amssymb, amsfonts}
\usepackage{empheq} %GET BOXED EQUATIONS

\usepackage{lmodern} % REMOVES SOME FONT WARNINGS

\usepackage{algorithmic,algorithm} % TO USE ALGORITHM ENVIRONMENT
\usepackage{todonotes}

\newcounter{mynotes}
\def\notes{1}
\newcommand{\gnote}[1]{\ifnum\notes=1{{\sf\color{blue} [Gopi: #1]}}\fi}
\newcommand{\vnote}[1]{\ifnum\notes=1{{\sf\color{red} [Venkat: #1]}}\fi}

\declaretheorem[within=section]{theorem}
\declaretheorem[sibling=theorem]{corollary}

\declaretheorem[sibling=theorem]{lemma}
\declaretheorem[sibling=theorem]{claim}

\declaretheorem[sibling=theorem]{definition}

\declaretheorem[sibling=theorem]{proposition}

\declaretheorem[sibling=theorem]{example}

\newtheoremstyle{case}{}{}{}{}{}{:}{ }{}
\theoremstyle{case}

\newcommand{\R}{\mathbb{R}} % REALS
\newcommand{\N}{\mathbb{N}} % NATURAL NUMBERS
\newcommand{\F}{\mathbb{F}}
\renewcommand{\L}{\mathbb{L}}
\newcommand{\K}{\mathbb{K}}
\newcommand{\cB}{\mathcal B}

\newcommand{\cP}{\mathcal P}

\DeclarePairedDelimiter\ceil{\lceil}{\rceil} % CEIL
\DeclarePairedDelimiter\floor{\lfloor}{\rfloor} % FLOOR
\newcommand{\rk}{\mathrm{rank}} % RANK

\newcommand{\set}[1]{\{#1\}}

\renewcommand{\epsilon}{\varepsilon}

  \newcommand{\beq}{\begin{equation}}
  \newcommand{\eeq}{\end{equation}}
  \newcommand{\beqn}{\begin{equation*}}
  \newcommand{\eeqn}{\end{equation*}}
  \newcommand{\beqr}{\begin{eqnarray}}
  \newcommand{\eeqr}{\end{eqnarray}}
  \newcommand{\beqrn}{\begin{eqnarray*}}
  \newcommand{\eeqrn}{\end{eqnarray*}}
  \newcommand{\bmline}{\begin{multline}}
  \newcommand{\emline}{\end{multline}}
  \newcommand{\bmlinen}{\begin{multline*}}
  \newcommand{\emlinen}{\end{multline*}}

\usepackage[margin=1in]{geometry}

\usepackage[utf8]{inputenc}
\usepackage[T1]{fontenc}

\usepackage{mathptmx}
\usepackage{hyperref}
\usepackage[perpage]{footmisc}

\usepackage{array, boldline, makecell, booktabs}

\renewcommand{\ge}{\geqslant}

\renewcommand{\le}{\leqslant}

\newcommand\conj[2]{{}^{#2}{#1}}
\newcommand\Id{\mathrm{Id}}
\newcommand\charF{\mathrm{char}}
\title{Improved Maximally Recoverable LRCs using Skew Polynomials}
\author{Sivakanth Gopi\thanks{Microsoft Research. Email: \texttt{sigopi@microsoft.com}.}\\
\and    Venkatesan Guruswami\thanks{University of California, Berkeley and Simons Institute for the Theory of Computing. Email: \texttt{venkatg@berkeley.edu}. Research supported in part by NSF grants CCF-1563742 and CCF-1814603, and CCF-2210823, and a Simons Investigator Award. Most of this work was done when the author was with the Computer Science Department at Carnegie Mellon University.}}
\date{}
\begin{document}

\maketitle
\thispagestyle{empty}

\begin{abstract}
An $(n,r,h,a,q)$-Local Reconstruction Code (LRC) is a linear code over $\mathbb{F}_q$ of length $n$, whose codeword symbols are partitioned into $n/r$ local groups each of size $r$. Each local group satisfies `$a$' local parity checks to recover from `$a$' erasures in that local group and there are further $h$ global parity checks to provide fault tolerance from more global erasure patterns. Such an LRC is Maximally Recoverable (MR), if it offers the best blend of locality and global erasure resilience---namely it can correct all erasure patterns whose recovery is information-theoretically feasible given the locality structure (these are precisely patterns with up to `$a$' erasures in each local group and an additional $h$ erasures anywhere in the codeword).

Random constructions can easily show the existence of 
MR LRCs over very large fields, but a
major algebraic challenge is to construct MR LRCs, or even show their
existence, over smaller fields, as well as understand inherent
lower bounds on their field size. We give an explicit construction of
$(n,r,h,a,q)$-MR LRCs with field size $q$ bounded by
$\left(O\left(\max\{r,n/r\}\right)\right)^{\min\{h,r-a\}}$. This significantly
improves upon known constructions in many practically relevant parameter
ranges. Moreover, it matches the lower bound
from~\cite{gopi2020maximally} in an interesting range of parameters
where $r=\Theta(\sqrt{n})$, $r-a=\Theta(\sqrt{n})$ and $h$ is a fixed
constant with $h\le a+2$, achieving the optimal field size of
$\Theta_{h}(n^{h/2}).$

Our construction is based on the theory of skew
polynomials.  We believe skew polynomials should have further
applications in coding and complexity theory; as a small illustration
we show how to capture algebraic results underlying list decoding folded
Reed-Solomon and multiplicity codes in a unified way within this theory.

\end{abstract}
\newpage
\tableofcontents
\thispagestyle{empty}
\newpage
\setcounter{page}{1}

\section{Introduction}

We present an approach to construct Maximally Recoverable Local Reconstruction Codes (MR LRCs) based on the theory of skew polynomials. Our construction matches or improves the field size of MR LRCs for most parameter regimes. We now describe the motivation of MR LRCs in the context of coding for distributed storage, and then formally define them and describe our results.

\smallskip
In modern large-scale distributed storage systems (DSS), data is partitioned and stored in individual servers, each with a small storage capacity of a few terabytes. A server can crash any time losing all the data it contains. Less catastrophically, a server often tends to become temporarily unavailable either due to system updates, network bottlenecks, or being busy serving requests of other users. There are thus two design objectives for a DSS. The first one is to never lose user data in the event of crashes (or at least make it highly improbable). The second is to service user requests with low latency despite some servers becoming temporarily unavailable. As the simple approach of replicating data is prohibitive in terms of storage costs, erasure codes are employed in DSS. Using a Reed-Solomon code, if we add $n-k$ parity check servers to $k$ data servers, we can recover user data from any $k$ available servers. But as $k$ gets larger, this does not meet our second objective of servicing user requests with low latency. Local Reconstruction Codes (LRCs) were invented precisely for achieving both the objectives while still maintaining storage efficiency. These codes have \emph{locality} which means that for a small number of erasures, any codeword symbol can be recovered quickly based on a small number of other codeword symbols.  At the same time, they can also recover the missing codeword symbols in the unlikely event of a larger number of erasures (but can do so less efficiently). Locality in distributed storage was first introduced in~\cite{HCL,CHL}, but LRCs were first formally defined and studied in \cite{GHSY} and \cite{Dimakis_0}.  Suitably optimized LRCs have been implemented in several large scale systems such as Microsoft Azure~\cite{HuangSX} and Facebook~\cite{XOR_ELE}, leading to enormous savings in storage costs and improved system reliability.% We will now define them formally.

\smallskip
An $(n,r,h,a,q)$-LRC is a linear code over $\F_q$ of length $n$, whose codeword symbols are partitioned into $n/r$ local groups each of size $r$. The coordinates in each local group satisfy `$a$' local parity checks and there are further $h$ global parity checks that all the $n$ coordinates satisfy. The local parity checks are used to recover from up to `$a$' erasures in a local group by reading at most $r-a$ symbols in that local group. The $h$ global parities are used to correct more global erasure patterns which involve more than $a$ erasures in each local group. The parity check matrix $H$ of an $(n,r,h,a,q)$-LRC has the structure shown in Equation~\ref{fig:MRtopology_intro}. 
\begin{equation}\label{fig:MRtopology_intro}
H=
\left[
\begin{array}{c|c|c|c}
A_1 & 0 & \cdots & 0\\
\hline
0 &A_2 & \cdots & 0\\
\hline
\vdots & \vdots & \ddots & \vdots \\
\hline
0 & 0 & \cdots & A_g \\
\hline
B_1 & B_2 & \cdots & B_g \\
\end{array}
\right].
\end{equation}
Here $g=n/r$ is the number of local groups. $A_1,A_2,\dots,A_g$ are $a\times r$ matrices over $\F_q$ which correspond to the local parity checks that each local group satisfies. $B_1,B_2,\dots,B_g$ are $h\times r$ matrices over $\F_q$ and together they represent the $h$ global parity checks that the codewords should satisfy.

\smallskip
Equivalently, from an encoding point of view, an $(n,r,h,a,q)$-LRC is obtained by adding $h$ global parity checks to $k$ data symbols, partitioning these $k+h$ symbols into local groups of size $r-a$, and then adding `$a$' local parity checks for each local group. As a result we have $n=k + h + a\cdot \frac{k+h}{r-a}$ codeword symbols. This is shown in Figure~\ref{Fig:LRC}.
\begin{figure}[h]
\label{Fig:LRC}
\includegraphics[scale=0.45,trim={0cm 13cm 19cm 10cm},clip]{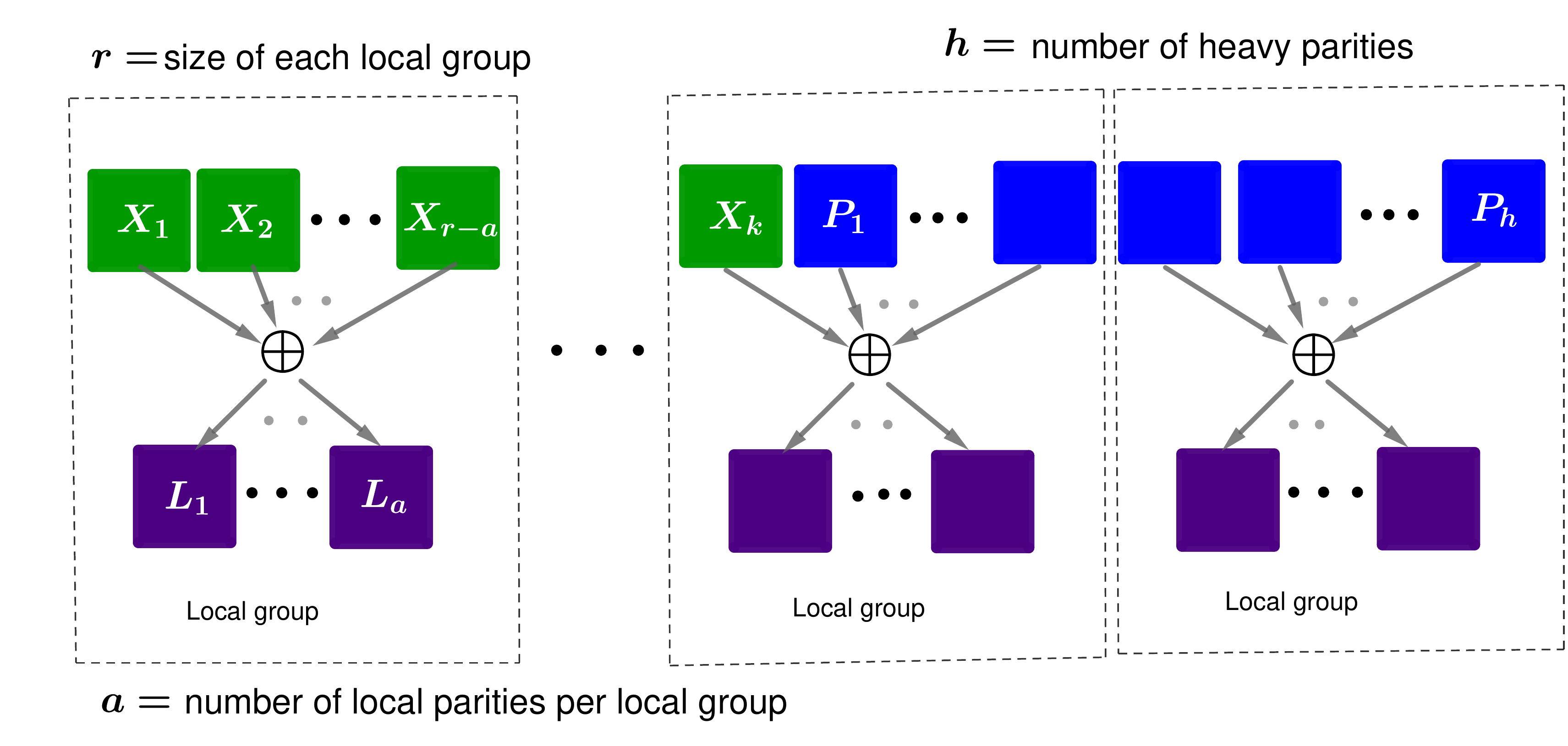}
\caption{An LRC with $k$ data symbols, $h$ heavy parities and `$a$' local parities per local group. The length of the code $n=k+h+a \cdot \frac{k+h}{r-a}$.}
\end{figure}

\smallskip
Information-theoretically, one can show that we can at best hope to correct an additional $h$ erasures distributed across global groups on top of the `$a$' erasures in each local group. LRCs which can correct all such erasure patterns which are information-theoretically possible to correct are called \emph{Maximally Recoverable (MR) LRCs}. The notion of maximal recoverability was first introduced by~\cite{CHL,HCL} and extended to more general settings in~\cite{GHJY}. But MR LRCs were specifically studied first by~\cite{BHH,Blaum} where they are called \emph{Partial-MDS (Maximum Distance Separable) codes}.

\begin{definition}\label{Def:MRLRC}
Let $C$ be an arbitrary $(n,r,h,a,q)$-local reconstruction code. We say that $C$ is maximally recoverable if:
\begin{enumerate}
	\item Any set of `$a$' erasures in a local group can be corrected by reading the rest of the $r-a$ symbols in that local group.
	\item Any erasure pattern $E\subseteq [n],$ $|E|=ga+h,$ where $E$ is obtained by selecting $a$ symbols from each of $g$ local groups and $h$ additional symbols arbitrarily, is correctable by the code $C.$
\end{enumerate}
\end{definition}

For a code $C$ with parity check matrix $H$, an erasure pattern $E$ is correctable iff the submatrix of $H$ formed by columns corresponding the coordinates in $E$ has full column rank.
Therefore, we have the following characterization of an MR LRC in terms of its parity check matrix.
\begin{proposition}
	\label{prop:MR_LRC_paritycheck}
	An $(n,r,h,a,q)$-LRC with parity check matrix given by $H$ from Equation~\ref{fig:MRtopology_intro} is maximally recoverable iff: 
\begin{enumerate}
	\item Each of the local parity check matrices $A_i$ are the parity check matrices of an MDS code, i.e., any $a$ columns of $A_i$ are linearly independent.
	\item Any submatrix of $H$ which can be formed by selecting $a$ columns in each local group and additional $h$ columns has full column rank.
\end{enumerate}
	
\end{proposition}

It is known that MR-LRCs exist over exponentially
large fields~\cite{GHSY}. This can be easily seen by instantiating the parity
check matrix $H$ from Equation~\ref{fig:MRtopology_intro} randomly
from an exponentially large field and verifying that the condition in
Proposition~\ref{prop:MR_LRC_paritycheck} is satisfied with high
probability by Schwartz-Zippel lemma. But codes deployed in practice
require small fields for computational efficiency, typically fields
such as $\F_{2^8}$ or $\F_{2^{16}}$ are preferred. Therefore a lot of
prior work focused on explicit constructions of MR LRCs over small
fields.

\subsection{Prior Work}

\noindent\textbf{Upper Bounds.} There are several known constructions of MR LRCs which are incomparable to each other in terms of the field size~\cite{GHJY,GYBS,GLX-ff,MK19,gopi2020maximally,Blaum,TPD,HY,GHKSWY,CK,BPSY}. Some constructions are better than others based on the range of parameters. Since there are too many parameters and there is no dominant regime of interest, it is helpful to think about what are the typical ranges of parameters that are useful in deployments of MR LRCs in practice.

\medskip\noindent\textbf{Parameter ranges useful in practice.} One should think of the number of local groups ($g$) as a constant and $n$ as growing. So $r=n/g$ is growing linearly with $n$. Typical values of $g$ used in practice are $g=2,3,4$. The number of global parities ($h$) should also be thought of as a small constant and the number of local parities $a$ is usually $1$ or $2$. The length $n$ of the code can range from $14$ to $60$. For example, an early version of Microsoft's Azure storage used $(n=14,r=7,h=2,a=1)$-MR LRCs with $g=2$ local groups~\cite{HuangSX}. These choices are mostly guided by the need to maximize storage efficiency (rate of the code) while balancing durability and fast reconstruction. This is different from the parameters of interest from a theoretical point of view, where to get locality we set $r$ to be sublinear in $n$. 
%\vnote{Good place to mention that there is no single dominant regime that's of interest. Could say in practice, small $h$ and $g$ are relevant, as well as small $a$, perhaps even $a=1$? From locality viewpoint, small $r$ is interesting.}

A few of the important prior constructions that work for all parameter ranges are shown in Table~\ref{tab:upper_bounds}. 
%The first part shows constructions which work for all ranges of parameters and the second part shows constructions which work for some special cases. 
The first bound by~\cite{GYBS} is good when $r$ is close to $n$. The second bound by~\cite{GLX-ff} is better when $h\ll r\ll n$. The bound by~\cite{MK19} is better when $r-a \le h.$ The construction in~\cite{MK19} is also significantly different from the previous constructions and our construction is inspired by the construction in~\cite{MK19}. 

\renewcommand{\arraystretch}{1.5}
%The height of each row is set to 1.5 relative to its default height.
\begin{table}[!ht]
\label{tab:upper_bounds}
\centering
\begin{tabular}{ |c|c|c|}
\hline
Field size $q$ & \\
\hline
$O\bigl(r\cdot n^{(a+1)h-1}\bigr)$  & \cite{GYBS} \\
\hline  
% $\max\left(O(n/r), O(r)^{h+a}\right)^h$& \cite{GYBS}\\
% \hline
$\max\bigl(O(n/r), O(r)^{\min\{r,h+a\}}\bigr)^{\min\{h,g\}}$&  \cite{GLX-ff}\\
\hline
$\left(O\bigl(\max\{n/r,r\}\bigr)\right)^{r-a}$ & \cite{MK19}\\
\hline
% \hline
% $O_r\left(n^{\ceil{(h-1)(1-1/2^r)}}\right)$ when $a=1$ and $r=O(1)$& \cite{GHJY}\\
% \hline
% $O(r)$ when $h=0$ or $h=1$& \cite{BHH}\\
% \hline
% $O(n)$ when $h=2$ & \cite{gopi2020maximally}\\
% \hline
% %$O(n^3)$ when $h=3$ & \cite{gopi2020maximally}\\
% %\hline
% $\widetilde{O}(n)$ when $h=3,a=1,r=3$ & \cite{gopi2020maximally}\\
% \hline
\end{tabular}
\caption{Table showing the best known upper bounds on the field size of $(n,r,h,a,q)$-MR LRCs.}
\end{table}
%\vnote{Compress above table to first 3 rows; mention the $h \le 2$ constructions in text.}

\renewcommand{\arraystretch}{1}
In some special cases, there are better constructions. \cite{GHJY} construct MR LRCs over fields of size $O_r\left(n^{\ceil{(h-1)(1-1/2^r)}}\right)$ when $a=1$ and $r=O(1)$. In the special case when $h=2$, a construction over linear sized fields for all ranges of other parameters is given in~\cite{gopi2020maximally}.

%Finally, the bound in~\cite{GHJY} is best when $a=1$ and $h\le r=O(1)$ are constants (we note that the implicit constant hidden in $O_r(\cdot)$ has an exponential dependence in $r$). In the special case when $h=2$, a construction over linear sized fields for all ranges of other parameters is given in~\cite{gopi2020maximally}.

\medskip\noindent\textbf{Lower Bounds.} The best known lower bounds on the field size required for $(n,r,h,a,q)$-MR LRCs (with $g=n/r$ local groups) is from~\cite{gopi2020maximally} who show that for $h\ge 2$,
\begin{equation}\label{eqn:lowerbound}
q \ge \Omega_{h,a} \left( n\cdot r^\alpha \right) \text{ where } \alpha=\frac{\min\left\{a,h-2\ceil{h/g}\right\}}{\ceil{h/g}}.
\end{equation}
The lower bound~(\ref{eqn:lowerbound}) simplifies to 
\begin{equation}
	\label{eqn:lowerbound_special}
q\ge \Omega_{h,a}\left(nr^{\min\{a,h-2\}}\right)
\end{equation}
 when $g=n/r\ge h.$ 
When $2\le h \le \min\{a+2,g\}$, we have:
\begin{equation}
	\label{eqn:lowerbound_special_hsmall}
q\ge \Omega_h\left(\frac{n (r-a)^{h-1}}{r}\right).
\end{equation}
Note that the hidden constant in (\ref{eqn:lowerbound_special_hsmall}) only depends on $h$.

\subsection{Our Results}
\noindent We are now ready to present our main result.
\begin{theorem}[Main]
	\label{thm:main_construction}
	Let $q_0\ge \max\{g+1,r-1\}$ be any prime power where $g=n/r$ is the number of local groups. Then there exists an explicit $(n,r,h,a,q)$-MR LRC with $q=q_0^{\min\{h,r-a\}}$. Asymptotically, the field size satisfies
\begin{equation}
	\label{eqn:main_construction}
	q \le \left(O\bigl(\max\{r,n/r\}\bigr)\right)^{\min\{h,r-a\}}.	
\end{equation}
\end{theorem}

Our construction is better than (or matches) the first three bounds in Table~\ref{tab:upper_bounds} for \emph{all} parameter ranges. Moreover when $h$ is a fixed constant with $h \le a+2$ and $r=\Theta(\sqrt{n})$ and $r-a=\Theta(\sqrt{n})$, our construction matches the lower bound ~(\ref{eqn:lowerbound_special_hsmall}), achieving the optimal field size of $\Theta_{h}(n^{h/2}).$ This is the first non-trivial case (other than when $h=2$~\cite{gopi2020maximally}) where we know the optimal field size for MR LRCs.
\begin{corollary}
	Suppose $r=\Theta(\sqrt{n})$, $r-a=\Theta(\sqrt{n})$ and $h$ is a fixed constant independent of $n$ such that $h\le a+2$. Then the optimal field size of an $(n,r,h,a,q)$-LRC is $q=\Theta_{h}(n^{h/2}).$
\end{corollary}

%We note that our construction is worse compared to the constructions in the second part of Table~\ref{tab:upper_bounds} which work for some special setting of parameters.
We also remark that the $h$ that appears in the field size upper bound in Theorem~\ref{thm:main_construction} can be replaced with $h_{\mathrm{local}}$, if we only want to correct erasure patterns formed by erasing `$a$' erasures in each local group and $h$ additional erasures, which are distributed in such a way that no local group has more than $a+h_{\mathrm{local}}$ erasures in total.

%\vnote{Comment that $h$ in our field size bound can be replaced by $h_{\text{local}}$. Introduce this concept, and motivate it's relevance to practice, along with some estimate of number of patterns with $h$ global erasures it corrects.}

% \gnote{I talked to some product folks and they said that being unable to correct $a+h$ erasures in a local group is not good. Because smaller size erasure patterns are the most likely, we should be able to correct all patterns of size $a+h$ at least.}
% \vnote{Above remark is good enough in that case.}

%\vnote{Even if not interesting in MSR specific practical setting and we don't seell the model of $h_{\text{local}}$ per se, I think it is good to point this aspect of the construction/bounds anyway, as a clarifying remark. Perhaps this can added to the remark below about not protecting global parities with local checks.}

MR LRCs used in practice typically have only a small constant number of local groups i.e. $g=n/r$ is typically a small constant such as $g=2,3,4$~\cite{HuangSX} and the number of local parities $a=1$. We can further improve the construction from Theorem~\ref{thm:main_construction} in this important regime.

\begin{theorem}
	\label{thm:construction_a1}
	Suppose the number of local parities $a=1$ and $g=n/r$ is the number of local groups. Let $q_0 \ge g+1$ be any prime power and let $C_0$ be any $[r,r-s,d]_{\F_{q_0}}$-code such that its parity check matrix contains a full weight row and it has distance $d\ge \min\{h,r-1\}+2$.\footnote{Equivalently, the dual code $C_{in}^\perp$ has a full weight codeword.} Then there exists an explicit $(n,r,h,a=1,q)$-MR LRC with field size $q=q_0^{s-1}.$ Asymptotically, by instantiating $C_0$ with BCH codes, we obtain a field size of $$q \le \bigl(O(n)\bigr)^{\ceil{\min\{h,r-1\}(1-1/q_0)}}.$$
\end{theorem}

We also remark that our constructions can be easily modified to the variant of MR LRCs where the global parities are not protected by the local parity checks. Since we did not define this variant of MR LRCs in this paper, we omit these constructions.
% \begin{theorem}
% 	\label{thm:construction_a1}
% 	Suppose the number of local parities $a=1$ and $g=n/r$ is the number of local groups. Let $q_0 \ge g+1$ be any prime power and let $s$ be such that $q_0^s \ge r$. Then there exists an explicit $(n,r,h,a=1,q)$-MR LRC with field size $q=q_0^{s\ceil{\min\{h,r-1\}(1-1/q_0)}}.$ Asymptotically, the field size satisfies $$q \le \bigl(O(n)\bigr)^{\ceil{\min\{h,r-1\}(1-1/q_0)}}.$$
% \end{theorem}

%\gnote{Maybe we can replace $O(n)$ with $\max\{r,g\}$.}

\medskip\noindent\textbf{Related Work.} Shortly before we published our results, we learned that \cite{cai2020construction} have independently obtained a result analogous to Theorem~\ref{thm:main_construction} with a very similar construction. They construct $(n,r,h,a,q)$-MR LRCs with a field size of 
\begin{equation}
	\label{eqn:caietal_construction}
q=\left(O\bigl(\max\{r,n/r\}\bigr)\right)^h.
\end{equation}
Compared to this, we have a $\min\{h,r-a\}$ in the exponent in our field size bound \eqref{eqn:main_construction}. The construction in the independent work~\cite{cai2020construction} is very similar to ours, we get $\min\{h,r-a\}$ in the exponent by being more careful in our analysis.

Soon after \cite{cai2020construction}, two more constructions of MR LRCs were published by~\cite{UMP-2020} with the following field sizes:
\begin{align}
q &\le \left(\max\Bigl\{(2r)^{r-a},\frac{g}{r}\Bigr\}\right)^{\min\{h,\floor{g/r}\}}, \label{eqn:martinez20_construction_1} \\
q &\le (2r)^{r-a}\left(\left\lfloor\frac{g}{r}\right\rfloor+1\right)^{h-1}. \label{eqn:martinez20_construction_2}
\end{align}
The constructions in (\ref{eqn:martinez20_construction_1}) and
(\ref{eqn:martinez20_construction_2}) are incomparable to our
construction in (\ref{eqn:main_construction}). For example when
$r=O(1)$, the construction (\ref{eqn:martinez20_construction_2})
achieves $O(n)^{h-1}$ field size, whereas our construction achieves
$O(n)^{\min\{h,r-a\}}$ field size. In the regime when
$r=\Theta(\sqrt{n})$ and $r-a=\Theta(\sqrt{n})$ and $h\le a+2$ is a fixed constant, our construction achieves the
optimal field size of $\Theta_{h}(n^{h/2})$, whereas the constructions
from \cite{UMP-2020} require fields of size $n^{\Theta(\sqrt{n})}.$

% \vnote{Skip this open question?}
% Despite all these constructions, a particularly interesting setting of parameters, which remains challenging is the case when $h=O(1)$ and $r-a=n^{o(1)}$. The lower bound~(\ref{eqn:lowerbound_special_hsmall}) only shows that $q=\Omega_h(n^{1+o(1)})$ whereas all the existing constructions need $q\gtrsim_h n^{h-1-o(1)}$. 

% \begin{openquestion}
% 	When $h=O(1)$ and $r=n^{o(1)}$, do there exist MR LRCs with field size $q\le n^{1+o(1)}$?
% \end{openquestion}

\subsection{Our Techniques}
%
%\vnote{Describe the construction, for example how the $B_i$'s are picked in \eqref{fig:MRtopology_intro}, as Moore matrices from different conjugacy classes. Lead into skew polynomials as a way to argue linear (in)-dependencies between columns of these matrices.}

%\gnote{I tried to add some more detail. Feel free to add more.} \vnote{Looks good to me except for two small/local comments I left below. Also made some small wording changes.}
Our constructions are based on the theory of skew polynomials and is
inspired by the construction from~\cite{MK19}. Skew polynomials are a
non-commutative generalization of polynomials, but they retain many of
the familiar and important properties of polynomials. Just as
Reed-Solomon codes are constructed using the fact that a degree $d$
polynomial can have at most $d$ roots, our codes will use an analogous
theorem that a degree $d$ skew polynomial can have at most $d$ roots
\emph{when counted appropriately} (see
Theorem~\ref{thm:fundamentalthm_roots_skewpolynomials}). Unlike the
roots of the usual degree $d$ polynomials which do not have any
structure, the roots of degree $d$ skew polynomials have an
interesting linear-algebraic structure which we exploit in our
constructions. The roots in $\F_{q^m}$ of a degree $d$ skew polynomial over $\F_{q^m}$ can be partitioned into \emph{conjugacy classes} such that the roots in each conjugacy class form a subspace over the base field $\F_q$. Moreover the sum of dimensions of these subspaces across conjugacy classes is at most $d$.

To exploit this root structure of skew polynomials in an MR LRC
construction, we associate each local group with a conjugacy class,
and the matrices $B_i$ in \eqref{fig:MRtopology_intro} are chosen so
that $\lambda^T B_i$ is the evaluation of a skew polynomial of degree
$d$ (with coefficients given by $\lambda$) over different points in
the same conjugacy class. 
% Because of the linear-algebraic structure of
% roots, these matrices $B_i$ look like Moore matrices. \vnote{Should we
%   define or point to Moore matrices as it won't be familiar to many?}
Across different local groups, we automatically get linear independence of
columns of matrices $B_1,B_2,\dots,B_g$ as these are associated with different conjugacy classes. Inside each local group, to
argue linear independence, the local parities $A_i$ will be chosen as
a Vandermonde matrix over the base field $\F_q$ (we can choose all the $A_i$'s to be equal), and the $B_i$ will be chosen
carefully to combine well with the Vandermonde matrix $A$ (see Equations~\eqref{eq:A-def}, \eqref{eq:beta-def}, \eqref{eq:choice-of-B}). 
%\vnote{Can the specific choice of $B_i$ be illuminated at some high level? I added pointers to the Eqns. defining these.}
%\gnote{I added the following line.} 
In particular, we choose $B_i$ so that the $(a+m)\times r$ matrix formed by adding the first row of $B_i$ with entries in $\F_{q^m}$ (but interpreted as an $m\times r$ matrix over the base field $\F_q$) to $A_i$ is an MDS matrix. This allows us to argue that any $a+m$ erasures in that local group can be corrected and we choose $m=\min\{h,r-a\}$. This is also the main difference between our work and \cite{MK19}, which is also implicitly
based on skew polynomials. 

In this paper, we make this connection
explicit in the hope that the theory of skew polynomials will lead to
further developments in the constructions of MR LRCs and coding theory
more broadly. As an illustration, in
Appendix~\ref{sec:wronskian-moore} we show how skew polynomials can give an explanation of algebraic results concerning (generalizations of) Wronskian and Moore matrices that have recently been used in the context of list decoding algorithms for folded Reed-Solomon and univariate multiplicity codes~\cite{GW13}, rank condensers~\cite{FS12,FSS14,FG15}, and subspace designs~\cite{GK-combinatorica,GXY-tams}. We also reproduce a construction of maximum sum-rank distance (MSRD) codes due to~\cite{Martinez18} using the framework of skew polynomials in Appendix~\ref{sec:MSRD}.  Skew polynomials have also been explicitly used before to define skew Reed-Solomon codes in \cite{boucher2014linear}.
Readers familiar with the theory of skew polynomials or who directly want to get to the construction can skip most of the preliminaries in Section~\ref{sec:prelim} except for Section~\ref{sec:vandermonde}.

%\vnote{I commented out below paragraph as it was a bit repetitive with previous one and added a bit to above paragraph.}
%Skew polynomials have been used directly and indirectly in coding theory before. As discussed in Appendix~\ref{sec:wronskian-moore}, folded Reed-Solomon codes and multiplicity codes can be thought of as special cases of skew polynomial based codes. \cite{boucher2014linear} used skew polynomials explicitly to define skew Reed-Solomon codes. \cite{Martinez18} used skew Reed-Solomon codes to construct maximum sum-rank codes, see Appendix~\ref{sec:MSRD} for the construction.

%\todo[inline]{Need to find out what \cite{UMP-2020} does and cite it also?}

\section{Preliminaries}
\label{sec:prelim}
{}\subsection{Skew polynomial ring}
%\vnote{Say that we give all background needed to read, but a more gentle presentation with proofs and examples is in the full version.}
Skew polynomials generalize polynomials while inheriting many of the nice properties of polynomials. Skew polynomials can be defined over division rings\footnote{Rings where every non-zero element has a multiplicative inverse, but multiplication may not be commutative.} and most of the results about skew polynomials are true in this more general setting. It is known that every finite division ring is a field. Since we will only work with skew polynomial rings defined over fields, we will only define them over fields for simplicity. Most of the theory of skew polynomials presented here is from~\cite{LamL88,Lam85}, but we reprove the main results in a more accessible way. Skew polynomials were first defined by Ore~\cite{ore1933} in 1933 where it was shown that they are the unique non-commutative generalization of polynomials which satisfy (1) associativity (2) distributivity on both sides and (3) the fact that the degree of product of two polynomials is the sum of their degrees. 

\noindent Let $\K$ be a field. We will first define the key concepts of `endomorphism' and `derivation'.
\begin{definition}[Endomorphism]
	A map $\sigma:\K\to \K$ is called an endomorphism if:
	\begin{enumerate}
          \itemsep=0ex
		\item $\sigma$ is a linear map i.e. $\sigma(a+b)=\sigma(a)+\sigma(b)$ for all $a,b\in K$ and
		\item $\sigma(ab)=\sigma(a)\sigma(b)$ for all $a,b\in K$.
	\end{enumerate}
\end{definition}
% \begin{example}
% 	\begin{enumerate}
% 		\item If $\K=\F_{q^m}$, then $\sigma(x)=x^q$ is an endomorphism.
% 		\item If $\K=\F(x)$ is the field of rational functions and $\gamma\in \F^*$, then $\sigma(f(x))=f(\gamma x)$ is an endomorphism.
% 	\end{enumerate}
% \end{example}

For example, if $\K=\F_{q^m}$, then $\sigma(x)=x^q$ is an endomorphism called the Frobenius endomorphism. If $\K=\F(x)$ is the field of rational functions and $\gamma\in \F^*$, then $\sigma(f(x))=f(\gamma x)$ is an endomorphism.

\begin{definition}[Derivation]
	A map $\delta:\K \to \K$ is called a $\sigma$-derivation if:
	\begin{enumerate}
                    \itemsep=0ex
		\item $\delta$ is a linear map i.e. $\delta(a+b)=\delta(a)+\delta(b)$ for all $a,b\in K$ and
		\item $\delta(ab)=\sigma(a)\delta(b)+\delta(a)b$ for all $a,b\in K$.
	\end{enumerate}
\end{definition}

We will now define the skew polynomial ring.
\begin{definition}[Skew polynomial ring]
	Let $\sigma$ be an endomorphism of $\K$ and $\delta$ be a $\sigma$-derivation. The skew polynomial ring in variable $t$, denoted by $\K[t;\sigma,\delta]$, is a non-commutative ring of skew polynomials in $t$ of the form $\{\sum_{i=0}^d a_i t^i: d\ge 0, a_i\in \K\}$ (where we always write the coefficients to the left). Degree of a polynomial $f(t)=\sum_i a_i t^i$, denoted by $\deg(f)$, is the largest $d$ such that $a_d \ne 0.$\footnote{We will define the degree of the zero polynomial to be $\infty.$} Addition in $\K[t;\sigma,\delta]$ is component wise. But multiplication is distributive and done according to the following rule:
	\begin{equation}
		\label{eqn:multiplication_rule}
			\text{For }a\in \K,\ t\cdot a=\sigma(a)t+\delta(a).
	\end{equation}
\end{definition}
To multiply $f(t)g(t)$, we can first use distributivity to get $f(t)g(t)=\sum_{ij} f_it^i \cdot g_j t^j$ where $f_i,g_j\in \K$ are coefficients of $f,g$ respectively. Then we use the rule~(\ref{eqn:multiplication_rule}) for $i$ times to move the coefficient $g_j$ to the left of $t^i$. This multiplication turns out to be associative, but may not be commutative. Also $\deg(f\cdot g)=\deg(f)+\deg(g)$. Therefore the skew polynomial ring has no zero divisors. We will now give some examples of skew-polynomials.

The simplest derivation is the zero map i.e. $\delta(a)=0$ for all $a\in \K$.  In this case, the skew polynomial ring is denoted by $\K[t;\sigma]$ and is said to be of endomorphism type. Skew polynomials are interesting even in this case, and in fact the constructions in this paper only use skew polynomials with $\delta\equiv 0$. So the reader can imagine that the derivation is the zero map on a first reading. We include the general case to discuss the applications of skew polynomials to coding and complexity theory later in Appendix~\ref{sec:wronskian-moore} and in the hope that skew polynomial rings with non-zero derivations will find applications in future. For more interesting examples of skew polynomial rings, see Appendix~\ref{sec:examples_skewpolyring}

We will now collect some simple facts about skew polynomials rings. Let $\K[t;\sigma,\delta]$ be a skew polynomial ring.
\begin{lemma}[\cite{LamL88}]
	$t^na=\sum_{i=0}^n f_i^n(a) t^i$ where $f_0^n=\delta^n,\ f_1^n=\delta^{n-1}\sigma + \delta^{n-2}\sigma \delta +\dots + \sigma \delta^{n-1},\dots, f_n^n=\sigma^n$ are linear maps.
\end{lemma}
It turns out that the skew polynomial ring has Euclidean algorithm for right division.
\begin{lemma}[Euclidean algorithm for right division~\cite{LamL88}]
	For every two polynomial $f,g\in \K[t;\sigma,\delta]$, there exist unique polynomials $q(t),r(t)$ such that $f=q\cdot g+r$ where $\deg(r)<\deg(g)$ or $r=0.$
\end{lemma}
This brings us to the most important definition about skew polynomial rings. In the usual polynomial world, we can define the evaluation of a polynomial $f(t)=\sum_i f_i t^i$ at $t=a$ as $\sum_i f_i a^i.$ With this definition, it is true that $f(t)=q(t)(t-a)+f(a).$ But for skew polynomials, these two notions of evaluation differ with each other and the right definition is the second one.

\begin{definition}[Evaluation]
	The evaluation of a polynomial $f\in \K[t;\sigma,\delta]$ at a point $a\in \K$, denoted by $f(a)$, is defined as the remainder obtained when we divide $f$ by $t-a$ on the right i.e. $f(t)=q(t)(t-a)+f(a).$
\end{definition}
Note that evaluation is a linear map i.e. $(f+g)(a)=f(a)+g(a).$ But it is not always true that $(fg)(a)=f(a)g(a)$. We will see shortly how to compute $(fg)(a)$.
The evaluation map can be expressed using ``power functions", which are the evaluations of monomials of the form $t^i.$

\begin{definition}[Power functions]
	\label{def:power_functions}
	The power functions are defined inductively as follows. For every $a \in \K$
	\begin{enumerate}
          \itemsep=0ex
		\item $N_0(a)=1$ and
		\item $N_{i+1}(a)=\sigma(N_i(a))a+\delta(N_{i}(a)).$
	\end{enumerate}
\end{definition}

When $\delta\equiv 0$, we have $\N_i(a)=\sigma^{i-1}(a)\sigma^{i-2}(a)\cdots \sigma(a) a$. Additionally if $\sigma\equiv \Id$, then $N_i(a)=a^i$ which explains the terms ``power functions".

\begin{lemma}
	\label{lem:evaluation}
	Let $f=\sum_i f_i t^i$. Then $f(a)=\sum_i f_i N_i(a).$
\end{lemma}
\begin{proof}
	It is easy to prove by induction that evaluation of $t^i$ at $a$ is $N_i(a)$. The general claim follows by linearity of evaluation.
\end{proof}

We now come to the problem of evaluating $(fg)(a).$ For this, it is useful to define the notion of \emph{conjugates}, which play a big role in this theory.
\subsection{Conjugation and Product Rule}
\begin{definition}[Conjugation]
	Let $a\in \K$ and $c\in \K^*$. We define the $c$-conjugate of $a$, denoted by $\conj a c$, as $$\conj a c = \sigma(c)ac^{-1}+\delta(c)c^{-1}.$$ We say that $b$ is a conjugate of $a$ if there exists some $c\in \K^*$ such that $b=\conj a c.$
\end{definition}

We have the following lemma which shows that conjugacy is an equivalence relation, we prove it in Appendix~\ref{sec:missing_proofs}.
\begin{lemma}
	\label{lem:conjugacy_equivalence}
	\begin{enumerate}
                    \itemsep=0ex
		\item $\conj{(\conj a c)}{d}=\conj{a}{dc}$
		\item Conjugacy is an equivalence relation, i.e., we can partition $\K$ into conjugacy classes where elements in each part are conjugates of each other, but elements in different parts are not conjugates.
	\end{enumerate}
\end{lemma}
% \begin{proof}
% 	(1) follows easily from the definition of conjugation and the using the fact that $\delta(cd)=\sigma(c)\delta(d)+\delta(c)d.$
% 	\begin{align*}
% 		\conj{(\conj{a}{c})}{d}&= \sigma(d)\cdot \conj{a}{c} \cdot d^{-1}+\delta(d)d^{-1}\\
% 		&= \sigma(d)(\sigma(c)ac^{-1}+\delta(c)c^{-1})d^{-1}+\delta(d)d^{-1}\\
% 		&= \sigma(dc)ac^{-1}d^{-1}+\sigma(d)\delta(c)c^{-1}d^{-1}+\delta(d)d^{-1}\\
% 		&= \sigma(dc)a(dc)^{-1}+(\sigma(d)\delta(c)+\delta(d)c)c^{-1}d^{-1}\\
% 		&= \sigma(dc)a(dc)^{-1}+\delta(dc)(dc)^{-1}\\
% 		&=\conj{a}{dc}.
% 	\end{align*}
% 	We now prove (2). Suppose $a$ is a conjugate of $b$, i.e., $a=\conj{b}{x}$ for some $x\in \K^*.$ Then $\conj{a}{x^{-1}}=\conj{(\conj{b}{x})}{x^{-1}}=\conj{b}{x^{-1}x}=b.$ Therefore $b$ is a conjugate of $a.$ Suppose $a$ is a conjugate of $b$, with $a=\conj{b}{x}$, and $c$ is a conjugate of $b$, with $b=\conj{c}{y}$. Then $a=\conj{b}{x}=\conj{(\conj{c}{y})}{x}=\conj{c}{xy}.$ So $a$ is a conjugate of $c.$
% \end{proof}
%It is easy to check that $\conj{(\conj a c)}{d}=\conj{a}{dc}$.\footnote{This explains the choice of left superscript for division rings. When we are working over fields, $dc=cd$ and therefore $\conj{a}{dc}=\conj{a}{cd}$. So the order doesn't matter.} Therefore conjugacy is an equivalence relation. It partitions $\K$ into equivalence classes. 
So $\K$ will get partitioned into conjugacy classes. To understand the structure of each conjugacy class, we need the notion of \emph{centralizer}.

 \begin{definition}[Centralizer]
 	The centralizer of $a\in\K$ is defined as:
 	$$\K_a=\{c\in \K^*: \conj{a}{c}=a\}\cup \set{0}.$$
 \end{definition}
The following lemma shows that centralizers are subfields, we prove it in Appendix~\ref{sec:missing_proofs}.
\begin{lemma}
\label{lem:centralizer_subfield}
\begin{enumerate}
            \itemsep=0ex
	\item $\K_a$ is a subfield of $\K.$\footnote{When $\K$ is a division ring, $\K_a$ will be a sub-division ring of $\K.$}
	\item If $a,b\in \K$ are conjugates, then $\K_a=\K_b$. \footnote{When $\K$ is a division ring and not a field, we have $\K_{(\conj{a}{x})}=x\K_ax^{-1}$.}
\end{enumerate}
\end{lemma}
% \begin{proof}
% 	(1) Let $x,y\in \K_a\setminus \{0\}$ i.e. $\conj{a}{x}=\conj{a}{y}=a$. Then
% 	\begin{align*}
% 		\conj{a}{x+y}(x+y)&=\sigma(c+d)a + \delta(c+d)\\
% 		&=\sigma(c)a+\sigma(d)a + \delta(c)+\delta(d)\\
% 		&=\conj{a}{c}c+\conj{a}{d}d\\
% 		&=ac+ad=a(c+d).
% 	\end{align*}
% 	Therefore $\conj{a}{x+y}=a$. Also $\conj{a}{yx}=\conj{(\conj{a}{x})}{y}=a.$ And finally $\conj{a}{x^{-1}}=\conj{(\conj{a}{x})}{x^{-1}}=\conj{a}{x^{-1}x}=a.$

% 	(2) Suppose $b=\conj{a}{d}$ and let $c\in \K_a.$Then $\conj{b}{c}=\conj{(\conj{a}{d})}{c}=\conj{a}{cd}=\conj{a}{dc}=\conj{(\conj{a}{c})}{d}=\conj{a}{d}=b.$ Therefore $\K_a\subset \K_b$. By symmetry, $\K_b\subset \K_a.$
% \end{proof}
Because of the above lemma, we can associate a centralizer subfield to each conjugacy class. 
\begin{example}
	\label{example:Frobenious}
	Let $\K=\F_{q^m}$, $\sigma(a)=a^q$ and $\delta\equiv 0$. Then $\conj{a}{c}=c^{q-1}a$. Suppose $\gamma$ is a generator for $\F_{q^m}^*$. There are $q$ equivalence classes, $E_{-1},E_0,E_1,\dots,E_{q-2}$, where $E_\ell=\{\gamma^i :i \equiv \ell \mod (q-1).\}$ and $E_{-1}=\{0\}.$ The centralizer of an element $a\in \K^*$ is $$\K_a=\set{c: c^{q-1}a=a}\cup \set{0}=\set{c:c^{q-1}=1}\cup \set{0}=\F_q.$$ Therefore the centralizer of every non-zero element is $\F_q$ and the centralizer of $0$ is $\K_0=\K.$
\end{example}

% \begin{example}
% 	In particular, if $\K=\F_{q^m}$ and $\sigma(x)=x^q$, then $\delta(x)=\lambda(x^q-x)$ for some $\lambda\in \K$ is the only possible derivation. Then $\conj{a}{c}=c^{q-1}a+\lambda(c^{q-1}-1)$.
% 	Suppose $\gamma$ is a generator for $\F_{q^m}^*$. There are $q-1$ equivalence classes, $E_{-1},E_0,E_1,\dots,E_{q-2}$, where $E_\ell=\{\gamma^i - \lambda :i \equiv \ell \mod (q-1).\}$ and $E_{-1}=\{-\lambda\}.$ The centralizer of an element $a\in \K\setminus\{-\lambda\}$ is $$\K_a=\set{c: c^{q-1}a+\lambda(c^{q-1}-1)=a}\cup \set{0}=\set{c:c^{q-1}=1}\cup \set{0}=\F_q.$$ Therefore the centralizer of every element other than $-\lambda$ is $\F_q$ and the centralizer of $-\lambda$ is $C(-\lambda)=\K.$ So this is very similar to Example~\ref{example:Frobenious}. In fact one can show that the skew polynomial rings are isomorphic for all values of $\lambda.$
% \end{example}

 We will now show how to evaluate $(fg)(a)$ using conjugation which plays a key role. The proof of this really important lemma is given in Appendix~\ref{sec:missing_proofs}.

\begin{lemma}[Product evaluation rule~\cite{Lam85,LamL88}]
	\label{lem:product_evaluation}
	If $g(a)=0$, then $(fg)(a)=0$. If $g(a)\ne 0$ then $$(fg)(a)=f\left(\conj{a}{g(a)}\right)g(a).$$
\end{lemma}

Using the product rule, one can prove an interpolation theorem for skew polynomials just like ordinary polynomials. For any $A\subset \K$ be of size $n$, there exists a non-zero degree $\le n$ skew polynomial $f\in \K[t;\sigma,\delta]$ which vanishes on $A$~\cite{LamL88}. We will later need the following lemma.
\begin{lemma}
	\label{lem:sum_of_conjugates}
	Let $f$ be any skew polynomial. Fix some $a \in \K.$ Then $D_{f,a}(y)=f(\conj{a}{y})y$ is an $\K_a$-linear map from $\K\to \K$.
\end{lemma}
\begin{proof}
	%Linearity follows from the fact that $N_n(\conj{a}{y})y=\sum_{i=0}^n f_i^n(y)N_i(a)$ and the fact that $f_i^n$ are linear functions.

	Linearity follows since $f(\conj{a}{y})y$ is equal to the evaluation of the polynomial $f(t)y$ at $a$ by Lemma~\ref{lem:product_evaluation}. And clearly the evaluation is linear in $y.$
	$\K_a$-linearity follows since $\forall c\in \K_a$, $$D_{f,a}(yc)=f(\conj{a}{yc})yc = f(\conj{(\conj{a}{c})}{y})yc=f(\conj{a}{y})yc=D_{f,a}(y)c. \qedhere $$

	% (2) This can be proved by induction, it is true for $i=1$. 
	% \begin{align*}
	% N_{i+1}(\conj{a}{y})y &= \sigma(N_i(\conj{a}{y}))\conj{a}{y}y + \delta(N_i(\conj{a}{y}))y\\
	%  &= \sigma(N_i(\conj{a}{y}))(\sigma(y)a+\delta(y)) + \delta(N_i(\conj{a}{y}))y\\
	%  & = \sigma(N_i(\conj{a}{y})y)a + \sigma(N_i(\conj{a}{y}))\delta(y) + \delta(N_i(\conj{a}{y}))y\\
	%  & = \sigma(N_i(\conj{a}{y})y)a + \delta(N_i(\conj{a}{y})y)\\
	%  & = \phi_a (N_i(\conj{a}{y})y) = \phi_a (\phi_a^i(y)) = \phi_a^{i+1}(y).
	%  \end{align*}
\end{proof}

% \begin{proposition}
% 	Let $\phi_a: \K \to \K$ be defined as $\phi_a(y)=\sigma(y)a+\delta(y)$. Then
% 	\begin{enumerate}
% 		\item $\phi_a$ is a linear map over the subfield $\K_a,$
% 		\item $\phi_a^i(y)=N_i(\conj{a}{y})y.$
% 	\end{enumerate}
% \end{proposition}
% Since, we chose to work over fields, we have the following commutativity relation for conjugation.
% \begin{lemma}
% 	When $\K$ is a field, we have $\conj{a}{cd}=\conj{a}{dc}$.
% \end{lemma}
% \begin{proof}
% 	Expand both sides using the definition of conjugation and use the fact that $\sigma(c)\delta(d)+\delta(c)d=\delta(cd)=\delta(dc)=\sigma(d)\delta(c)+\delta(d)c$.
% \end{proof}

\subsection{Roots of skew polynomials}
The most important and useful fact about usual polynomials is that a degree $d$ non-zero polynomial can have at most $d$ roots. It turns out that this statement is false for skew polynomials! A skew polynomial can have many more roots than its degree. But when counted in the right way, we can recover an analogous statement for skew polynomials. In this section, we will prove the ``fundamental theorem'' about roots of skew polynomials which shows that a degree $d$ skew polynomial cannot have more than $d$ roots when counted the right way. Before we state the fundamental theorem, let us try to understand the roots of a skew polynomial in the same conjugacy class. The following lemma shows that they form a vector space over a subfield of $\K.$

\begin{lemma}
	\label{lem:conjugate_roots_vectorspace}
	Let $f\in \K[t;\sigma,\delta]$ be a non-zero polynomial and fix some $a\in \K$ and let $\F=\K_a$ be the centralizer of $a$ (which is a subfield of $\K$). Define $V_f(a)=\{y\in \K^*: f(\conj{a}{y})=0\}\cup \set{0}$. Then $V_f(a)$ is a vector space over $\F$.
\end{lemma}
\begin{proof}
	For any $\lambda\in \F$ and $y\in V_f(a)$, $f(\conj{a}{\lambda y})=f(\conj{(\conj{a}{\lambda})}{y})=f(\conj{a}{y})=0$. Therefore $\lambda y\in V_f(a)$.
	If $y_1,y_2\in V_f(a)$ where $y_1+y_2\ne 0$, then by Lemma~\ref{lem:sum_of_conjugates}, $f(\conj{a}{y_1+y_2})=0$. Therefore $y_1+y_2\in V_f(a).$
\end{proof}

We are now ready to state the ``fundamental theorem'' about roots of skew polynomials, the proof appears in Appendix~\ref{sec:proof_fundamental_thm}.
\begin{theorem}[\cite{Lam85,LamL88}]
	\label{thm:fundamentalthm_roots_skewpolynomials}
	Let $f\in \K[t;\sigma,\delta]$ be a degree $d$ non-zero polynomial. Let $A$ be the set of roots of $f$ in $\K$ and let $A=\cup_i A_i$ be a partition of $A$ into conjugacy classes. Fix some representatives $a_i\in A_i$. Let $V_i=\{y: \conj{a_i}{y}\in A_i\}\cup \{0\}$ which is a linear subspace over $\F_i=\K_{a_i}$ by Lemma~\ref{lem:conjugate_roots_vectorspace}. Then $$\sum_i \dim_{\F_i}(V_i)\le d.$$
\end{theorem}
In particular, this implies that a non-zero degree $d$ polynomial can have roots in at most $d$ distinct conjugacy classes. And the dimension (over the centralizer subfield) of the subspace of roots in a single conjugacy class is at most $d.$

\subsection{Vandermonde matrix}
\label{sec:vandermonde}
\begin{definition}[Vandermonde matrix]
	Let $A=\set{a_1,\dots,a_n} \subset \K$. The $d\times n$ Vandermonde matrix formed by $A$, denoted by $V_d(a_1,\dots,a_n)$, is defined as:
	$$V_d(a_1,\dots,a_n)=
	\begin{bmatrix}
	N_0(a_1) & N_0(a_2) & \cdots & N_0(a_n)\\
	N_1(a_1) & N_1(a_2) & \cdots & N_1(a_n)\\
	\vdots & \vdots  &  & \vdots \\
	N_{d-1}(a_1) & N_{d-1}(a_2) & \cdots & N_{d-1}(a_n)
	\end{bmatrix}.$$
\end{definition}
 When the order of $a_1,a_2,\dots,a_n$ is not important, we sometimes denote $V_d(a_1,a_2,\dots,a_n)$ be $V_d(A).$
\noindent If $f(t)=\sum_{i=0}^{d-1} f_i t^i$ is a skew polynomial of degree at most $d-1$, then by Lemma~\ref{lem:evaluation}, 
\begin{equation}
\label{eqn:vandermonde_evaluation}
[f_0 f_1 \cdots f_{d-1}]\cdot V_d(a_1,a_2,\dots,a_n)=[f(a_1) f(a_2) \cdots f(a_n)].	
\end{equation}

 \begin{lemma}
 	\label{lem:rank_Vandermonde}
 	Let $A\subset \K$ of size $d$. Let $A=A_1\cup A_2 \cup \dots \cup A_r$ be the partition of $A$ into different conjugacy classes. Let $n_i=|A_i|$ and let $A_i =\{\conj{a_i}{c_{ij}}: j\in [n_i]\}$. Then $V_d(A)$ is full rank if for each $i\in [r]$, $\{c_{ij}: j\in [n_i]\}$ are linearly independent over the centralizer subfield $\K_{a_i}$.
 \end{lemma}
 \begin{proof}
 	If $V_d(A)$ is not full rank then there exists some non-zero row vector $[f_0\ f_1\ \dots\ f_{d-1}]$ such that $[f_0\ f_1\ \dots\ f_{d-1}]\cdot V_d(A)=0$. Therefore the non-zero skew polynomial $f(t)=\sum_{i=0}^{d-1} f_i t^i$, with degree at most $d-1$, has roots at all points of $A$. This violates Theorem~\ref{thm:fundamentalthm_roots_skewpolynomials}.
 \end{proof}

% \begin{lemma}
% 	\label{lem:Vandermonde_distinct_conjugacy}
% 	Let $a_1,\dots,a_d\in \K$ be in distinct conjugacy classes. Then $V_d(a_1,\dots,a_d)$ is full-rank.
% \end{lemma}
% \begin{proof}
% 	If not, then there exists a non-zero vector $(f_0,f_1,\dots,f_{d-1})\in \K^d$ such that $[f_0,f_1,\dots,f_{d-1}]\cdot V_d(a_1,\dots,a_d)=0$. By Equation~(\ref{eqn:vandermonde_evaluation}), this implies that the skew polynomial $f(t)=\sum_{i=0}^{d-1}f_it^i$ has $d$ roots in distinct conjugacy classes. This is a contradiction by Theorem~\ref{thm:fundamentalthm_roots_skewpolynomials}.
% \end{proof}

We will now see two corollaries of Lemma~\ref{lem:rank_Vandermonde} which are useful for our MR LRC construction.

\begin{corollary}
	 Let $\gamma\in \F_{q^m}^*$ be a generator of the multiplicative group. Let $d\le q-1$ and $\ell_1,\dots,\ell_d\in \set{0,1,2,\dots,q-2}$ be distinct. Then the following matrix $M$ is full rank.
	 $$
	 M=
	 \begin{bmatrix}
	 	1& 1 & \dots & 1\\
	 	\gamma^{\ell_1}& \gamma^{\ell_2}& \cdots & \gamma^{\ell_d}\\
	 	\gamma^{\ell_1(1+q)}& \gamma^{\ell_2(1+q)}& \cdots & \gamma^{\ell_d(1+q)}\\
	 	\vdots & \vdots &  &\vdots\\
	 	\gamma^{\ell_1(1+q+\dots+q^{d-2})}& \gamma^{\ell_2(1+q+\dots+q^{d-2})}& \cdots & \gamma^{\ell_d(1+q+\dots+q^{d-2})}\\
	 \end{bmatrix}
	 $$
\end{corollary}
\begin{proof}
	Let $\K=\F_{q^m}$, $\sigma(a)=a^q$ and $\delta\equiv 0$. Then $N_i(a)=a^{1+q+q^2+\dots+q^{i-1}}$. By Lemma~\ref{lem:rank_Vandermonde}, it is enough to show that $\ell_1,\dots,\ell_d$ fall in distinct conjugacy classes. This is shown in Example~\ref{example:Frobenious}.
\end{proof}
Note that when $m=1$, the matrix in the above corollary reduces to the usual Vandermonde matrix one is familiar with.
 
 % In general we would want to compute the rank of $V_n(a_1,\dots,a_n)$ for any given $a_1,\dots,a_n$. The following lemma generalizes Lemma~\ref{lem:Vandermonde_distinct_conjugacy}.

 % \begin{lemma}
 % 	\label{lem:rank_Vandermonde}
 % 	Let $A=\set{a_1,\dots,a_n}\subset \K$. Let $A=A_1\cup A_2 \cup \dots \cup A_r$ be the partition of $A$ into conjugacy classes. Then $\rk(V_n(A))=\sum_i \rk(V_n(A_i))$.
 % \end{lemma}
 % \begin{proof}
 % 	Follows from Theorem~\ref{thm:fundamentalthm_roots_skewpolynomials}.
 % \end{proof}

% \begin{lemma}
% 	$\rk(V_m(a_1,\dots,a_n))=\min(m,\rk(V_n(a_1,\dots,a_n)))$.
% \end{lemma}
% \begin{proof}
% 	TODO
% \end{proof}
% By the above lemmas, we reduced the problem to computing $\rk(V_n(A))$ when all elements of $A$ belong to the same conjugacy class. The following lemma shows how to compute this.

% \begin{lemma}
% 	\label{lem:Vandermonde_same_conjugacy}
% 	Let $a\in \K$ and $\F=\K_a$ which is a subfield of $\K$. Then for any $\set{c_1,\dots,c_n}\subset \K^*$, we have $$\rk(V_n(\conj{a}{c_1},\dots,\conj{a}{c_n}))=\dim_\F \linearspan_\F\{c_1,\dots,c_n\}.$$ In particular, $V_n(\conj{a}{c_1},\dots,\conj{a}{c_n})$ is full-rank iff $\{c_1,\dots,c_n\}$ are linearly independent over $\F$.
% \end{lemma}
% \begin{proof}
% 	Follows from Lemma~\ref{lem:conjugate_roots_dim}.
% \end{proof}

\begin{corollary}
	 Let $\gamma\in \F_{q^m}^*$ be a generator of the multiplicative group and let $\ell\in \set{0,1,\dots,q-2}$. Let $\beta_1,\dots,\beta_m\in \F_{q^m}$ be linearly independent over $\F_q$. Then the following matrix $M$ is full rank.
	 $$
	 M=
	 \begin{bmatrix}
	 	1& 1 & \dots & 1\\
	 	\gamma^{\ell}\beta_1^{q-1}& \gamma^{\ell}\beta_2^{q-1}& \cdots & \gamma^{\ell}\beta_m^{q-1}\\
	 	\gamma^{\ell(1+q)}\beta_1^{q^2-1}& \gamma^{\ell(1+q)}\beta_2^{q^2-1}& \cdots & \gamma^{\ell(1+q)}\beta_m^{q^2-1}\\
	 	\vdots & \vdots &  &\vdots\\
	 	\gamma^{\ell(1+q+\dots+q^{m-2})}\beta_1^{q^{m-1}-1}& \gamma^{\ell(1+q+\dots+q^{m-2})}\beta_2^{q^{m-1}-1}& \cdots & \gamma^{\ell(1+q+\dots+q^{m-2})}\beta_m^{q^{m-1}-1}
	 \end{bmatrix}
	 $$
\end{corollary}
\begin{proof}
	Let $\K=\F_{q^m}$, $\sigma(a)=a^q$ and $\delta\equiv 0$. Then $N_i(a)=a^{1+q+q^2+\dots+q^{i-1}}$. Let $a=\gamma^\ell$ then $M=V_m(\conj{a}{\beta_1},\dots,\conj{a}{\beta_m})$. Therefore $M$ is full rank by Lemma~\ref{lem:rank_Vandermonde}.
\end{proof}

\section{Skew polynomials based MR LRC constructions}
\label{sec:constructions_main}
Let us recall that an $(n,r,h,a,q)$-LRC admits a parity check matrix $H$ of the following form
\begin{equation}\label{fig:MRtopology}
H=
\left[
\begin{array}{c|c|c|c}
A_1 & 0 & \cdots & 0\\
\hline
0 &A_2 & \cdots & 0\\
\hline
\vdots & \vdots & \ddots & \vdots \\
\hline
0 & 0 & \cdots & A_g \\
\hline
B_1 & B_2 & \cdots & B_g \\
\end{array}
\right].
\end{equation}
Here $A_1,A_2,\cdots,A_g$ are $a\times r$ matrices over $\F_q$ which represent the local parity checks, $B_1,B_2,\cdots,B_g$ are $h\times r$ matrices over $\F_q$ which together represent the $h$ global parity checks. The rest of the matrix is filled with zeros. By Proposition~\ref{prop:MR_LRC_paritycheck}, $C$ is an MR LRC iff (1) any `$a$' columns of each matrix $A_i$ are linearly independent and (2) any submatrix of $H$ formed by selecting $a$ columns in each local group and any $h$ additional columns is full rank.

% Every matrix $\{A_i\}_{i\in [g]}$ is a parity check matrix of an $[r,r-a,a+1]$ MDS code. The bottom $h$ rows of $H$ serve to increase the code co-dimension from $ag$ to $ag+h$.

\subsection{Construction: Proof of Theorem~\ref{thm:main_construction}}
In this section, we will prove Theorem~\ref{thm:main_construction} by presenting a construction of MR LRCs over fields of size $q=O\left(\max(g,r)\right)^{\min\{h,r-a\}}.$
The construction presented here is inspired from~\cite{MK19}, where they achieve a field size of $O\left(\max(g,r)\right)^{r-a}$.%\footnote{The improvement comes from choosing $\beta_1,\dots,\beta_r$ carefully in our construction. Moreover~\cite{MK19} constructs a generator matrix for the code, whereas we construct a parity check matrix.}

Let ${q_0}\ge \max\{g+1,r\}$ be a prime power. Choose $\alpha_1,\alpha_2,\dots,\alpha_r\in \F_{q_0}$ to be distinct. Define
\begin{equation}
\label{eq:A-def}
A_\ell=
\begin{bmatrix}
	1 & 1 & \dots & 1\\
	\alpha_1 & \alpha_2 &\dots & \alpha_r\\
	\alpha_1^2 & \alpha_2^2 &\dots & \alpha_r^2\\
	\vdots & \vdots & & \vdots\\
	\alpha_1^{a-1} & \alpha_2^{a-1} &\dots & \alpha_r^{a-1}
\end{bmatrix}.
\end{equation}
Note that $A_1=A_2=\dots=A_g$.  Let $m=\min\{r-a,h\}$ and let $\gamma$ be a generator for $\F_{q_0^m}^*$. Our codes will be defined over the field $\F_q=\F_{q_0^m}$. Define $\beta_1,\beta_2,\dots,\beta_r\in \F_{q_0^m}$ as
\begin{equation}
\label{eq:beta-def}
\beta_i=
\begin{bmatrix}
	\alpha_i^{a}\\
	\alpha_i^{a+1}\\
	\vdots\\
	\alpha_i^{a+m-1}
\end{bmatrix},
\end{equation}
where we are expressing $\beta_i$ in some basis for $\F_{q_0^m}$ (which is a $\F_{q_0}$-vector space of dimension $m$). 
The improvement in our construction over \cite{MK19} comes from choosing $\beta_i$ carefully in our construction. In \cite{MK19}, $\beta_i$ are chosen independently of the local parity check matrix $A_i$ and they are chosen to satisfy $(r-a)$-wise independence over the base field $\F_{q_0}$. By choosing them carefully in combination with the local parity check matrix $A_i$, we only require $m=\min\{h,r-a\}$-wise independence of $\beta_1,\beta_2,\dots,\beta_r$. Moreover~\cite{MK19} constructs a generator matrix for the code, whereas we construct a parity check matrix. 

%\vnote{Emphasize the novelty in choice of $\beta_i$'s and in particular that they are not linear independent over $\F_q$ as in \cite{MK19}. Could merge this with the footnote, and move earlier, too.}
Define
\begin{equation}
\label{eq:choice-of-B}
B_\ell=
\begin{bmatrix}
	 	\beta_1& \beta_2 & \dots & \beta_{r}\\
	 	\gamma^{\ell}\beta_1^{{q_0}}& \gamma^{\ell}\beta_2^{{q_0}}& \cdots & \gamma^{\ell}\beta_{r}^{{q_0}}\\
	 	\gamma^{\ell(1+{q_0})}\beta_1^{q_0^2}& \gamma^{\ell(1+{q_0})}\beta_2^{q_0^2}& \cdots & \gamma^{\ell(1+{q_0})}\beta_{r}^{q_0^2}\\
	 	\vdots & \vdots &  &\vdots\\
	 	\gamma^{\ell(1+{q_0}+\dots+q_0^{h-2})}\beta_1^{q_0^{h-1}}& \gamma^{\ell(1+{q_0}+\dots+q_0^{h-2})}\beta_2^{q_0^{h-1}}& \cdots & \gamma^{\ell(1+{q_0}+\dots+q_0^{h-2})}\beta_{r}^{q_0^{h-1}}
	 \end{bmatrix}.
\end{equation}
To prove that the above construction is an MR LRC, we will use
properties of the skew field $\F_{q_0^m}[x;\sigma]$ where
$\sigma(a)=a^{q_0}$. We know that $\F_{q_0^m}$ will get partitioned
into ${q_0}-1$ conjugacy classes as shown in
Example~\ref{example:Frobenious}. If $\gamma\in \F_{q_0^m}^*$ is a
generator of $\F_{q_0^m}^*$, then
$\{1,\gamma,\gamma^2,\dots,\gamma^{{q_0}-2}\}$ fall in distinct
conjugacy classes. Intuitively, in the construction each local group
corresponds to one conjugacy class. This is possible since we chose
${q_0}\ge g+1.$ The stabilizer subfield of each conjugacy class is
$\F_{q_0}$ as shown in Example~\ref{example:Frobenious}. Therefore we
choose the matrices $B_i$ for local group $i$ as a (skew) Vandermonde
matrix where the evaluation points $\beta_1,\cdots,\beta_r$ are from
the conjugacy class of $\gamma^{i}$, but are linearly independent over
the stabilizer subfield $\F_{q_0}.$

\begin{claim}
\label{claim:MRLRC_main}
	The above construction is an MR LRC over fields of size ${q}=q_0^{\min\{h,r-a\}}.$
\end{claim}
\begin{proof}
	For a matrix $M$ and a subset $X$ of its columns, we will use $M(X)$ to denote the submatrix of $M$ formed by columns in $X.$ Given an erasure pattern $E$ of size $|E|=ag+h$, composed of $a$ erasures in each local group and $h$ additional erasures, we want to argue that the submatrix $H(E)$ is full rank. WLOG, assume that the $h$ additional erasures happen in local groups $1,2,\dots,t\in [g]$ for $t\le h.$ Let $E_i$ be the set of erasures that happen in the $i^{th}$ local group. Let $S_i\subset E_i$ be an arbitrary subset of size $|S_i|=a$ and let $T_i=E_i\setminus S_i.$ Note that $|T_i|\le m$ for all $i$. We need to show that $H(E)$ (which is an $(ag+h)\times (ag+h)$ matrix) is full rank where
	\begin{align*}
	H(E)=
	\left[
	\begin{array}{c|c|c|c}
	A_1(S_1\cup T_1) & 0 & \cdots & 0\\
	\hline
	0 &A_2(S_2\cup T_2) & \cdots & 0\\
	\hline
	\vdots & \vdots & \ddots & \vdots \\
	\hline
	0 & 0 & \cdots & A_g(S_g \cup T_g) \\
	\hline
	B_1(S_1\cup T_1) & B_2(S_2\cup T_2) & \cdots & B_g(S_g \cup T_g) \\
	\end{array}
	\right].
	\end{align*}
	Note that $A_1(S_1),A_2(S_2),\cdots, A_g(S_g)$ are $a\times a$ matrices of full rank. By doing column operations on $H(E)$, in each local group we can use the columns of $A_i(S_i)$ to remove the columns of $A_i(T_i)$. This results in the lower block $B_i(T_i)$ to change into a Schur complement as follows:
	\begin{align*}
		\left[
		\begin{array}{c|c}
			A_i(S_i) & A_i(T_i)\\
			\hline
			B_i(S_i) & B_i(T_i)\\ 
		\end{array}
		\right] \rightarrow 
		\left[
		\begin{array}{c|c}
			A_i(S_i) & 0\\
			\hline
			B_i(S_i) & B_i(T_i)-B_i(S_i)A_i(S_i)^{-1}A_i(T_i)\\ 
		\end{array}
		\right].
	\end{align*}
	 Note that $T_i=\phi$ for $i>t$.  So by doing row and column operations on $H(E)$, we can set it in a block diagonal form, where the diagonal blocks are given by $A_1(S_1),A_2(S_2),\dots,A_g(S_g)$ and one additional $h\times h$ block given by
	\begin{align*}
	C=
	\left[
	\begin{array}{c|c|c}
	B_1(T_1)-B_1(S_1)A_1(S_1)^{-1}A_1(T_1) & \cdots & B_t(T_t)-B_t(S_t)A_t(S_t)^{-1}A_t(T_t) \\
	\end{array}
	\right].
	\end{align*}
	Note that all the entries in $A(S_i)^{-1}A_i(T_i)$ are in the base field $\F_{q_0}.$ Also column operations on $B_i$ with $\F_{q_0}$ coefficients retain its structure with $\beta$'s replaced by their corresponding $\F_{q_0}$-linear combinations.
	Therefore by Lemma~\ref{lem:rank_Vandermonde}, it is enough to show that the following $t$ matrices $D_1,D_2,\dots,D_t$ are full rank:
	\begin{align*}
	D_i=
	\begin{bmatrix}
	\beta(T_i)-\beta(S_i)A_i(S_i)^{-1}A_i(T_i)
	\end{bmatrix}
	\end{align*}
	where $\beta=[\beta_1,\dots,\beta_r]$ is a $m\times r$ matrix over $\F_{q_0}$. Note that $[D_1|D_2|\dots|D_t]$ is just the first row of $C$ (with entries in $\F_{q_0^m}$) expressed as a matrix over $\F_{q_0}$. Consider following matrices given by 
	\begin{align*}
		F_i=\left[
		\begin{array}{c|c}
			A_i(S_i) & A_i(T_i)\\
			\hline
			\beta(S_i) & \beta(T_i)\\ 
		\end{array}
		\right]
	\end{align*}
	where each $F_i$ is of size $(a+m)\times (a+|T_i|)$. Each $F_i$ is a Vandermonde matrix by construction. Since $|T_i|\le m$, each $F_i$ is full rank. Now if we do column operations to get $F_i$ into block diagonal form we get:

	\begin{align*}
	\left[
	\begin{array}{c|c}
		A_i(S_i) & 0\\
		\hline
		\beta(S_i) & \beta(T_i)-\beta(S_i)A_i(S_i)^{-1}A(T_i)\\ 
	\end{array}
	\right]=
	\left[
	\begin{array}{c|c}
		A_i(S_i) & 0\\
		\hline
		\beta(S_i) & D_i\\ 
	\end{array}
	\right].
	\end{align*}
	This implies that $D_1,D_2,\dots,D_t$ are full rank over $\F_{q_0}$ which completes the proof.
\end{proof}

A slightly better construction which only requires $q_0 \ge \max\{g+1,r-1\}$ can be obtained by choosing
$$
A_\ell=
\begin{bmatrix}
1& \alpha_2^{m+a-1} & \alpha_3^{m+a-1} &\dots & \alpha_r^{m+a-1}\\
0& \alpha_2^{m+a-2} & \alpha_3^{m+a-2} &\dots & \alpha_r^{m+a-2}\\
\vdots&\vdots & \vdots & & \vdots\\
0&\alpha_2^{m+1} & \alpha_3^{m+1} &\dots & \alpha_r^{m+1}\\
0&\alpha_2^{m} & \alpha_3^{m} &\dots & \alpha_r^{m}\\
\end{bmatrix}
$$
and $\beta_1,\beta_2,\dots,\beta_r \in \F_{q_0}^m$ as:
\begin{align*}
\beta_1=
\begin{bmatrix}
	0\\
	\vdots\\
	0\\
	0
\end{bmatrix}
\text{ and }
\beta_i=
\begin{bmatrix}
	\alpha_i^{m-1}\\
	\vdots\\
	\alpha_i\\
	1
\end{bmatrix}
\text{ for } i\in \{2,3,\dots,r\}.	
\end{align*}

\subsection{Construction: Proof of Theorem~\ref{thm:construction_a1}}

\newcommand{\tH}{\tilde{H}}
When $a=1$ and $g$ is a fixed constant, we can improve the construction from the previous section using ideas from BCH codes.
Let $q_0\ge g+1$ be a prime power. Define
$$
A_\ell=
\begin{bmatrix}
	1 & 1 & \cdots & 1
\end{bmatrix}.
$$
Note that $A_1=A_2=\dots=A_g$. Let $H_{s\times r}$ be the parity check matrix of the $[r,r-s,d]_{\F_{q_0}}$-code $C_0$. By scaling the columns of $H$ and permuting the rows (which doesn't change the distance of $C_0$), we can assume that the first row of $H$ is $[1 1 \cdots 1]$. Let $\tH_{(s-1)\times r}$ be the submatrix of $H$ formed by removing the first row. Now define $\beta_1,\beta_2,\dots,\beta_r\in \F_{q_0}^s$ as the columns of $\tH$, i.e.,
$$\begin{bmatrix}\beta_1 & \beta_2 & \cdots & \beta_r\end{bmatrix}=\tH.$$

Here we are expressing $\beta_i$ in some basis for $\F_{q_0^{s-1}}$ (which is a $\F_{q_0}$-vector space of dimension $s-1$). Let $\gamma$ be a generator of $\F_{q_0^{s-1}}^*$.
Define $B_\ell$ as in (\ref{eq:choice-of-B}).
% $$
% B_\ell=
% \begin{bmatrix}
% 	 	\beta_1& \beta_2 & \dots & \beta_{r}\\
% 	 	\gamma^{\ell}\beta_1^{q_0}& \gamma^{\ell}\beta_2^{q_0}& \cdots & \gamma^{\ell}\beta_{r}^{q_0}\\
% 	 	\gamma^{\ell(1+q_0)}\beta_1^{q_0^2}& \gamma^{\ell(1+q_0)}\beta_2^{q_0^2}& \cdots & \gamma^{\ell(1+q_0)}\beta_{r}^{q_0^2}\\
% 	 	\vdots & \vdots &  &\vdots\\
% 	 	\gamma^{\ell(1+q_0+\dots+q_0^{h-2})}\beta_1^{q_0^{h-1}}& \gamma^{\ell(1+q_0+\dots+q_0^{h-2})}\beta_2^{q_0^{h-1}}& \cdots & \gamma^{\ell(1+q_0+\dots+q_0^{h-2})}\beta_{r}^{q_0^{h-1}}
% 	 \end{bmatrix}.
% $$ 

\begin{claim}
	\label{claim:MRLRC_main_a1}
	The above construction is an MR LRC over fields of size $q=q_0^{s-1}$.
\end{claim}
\begin{proof}
	The proof is analogous to the proof of Claim~\ref{claim:MRLRC_main}. Let $m=\min\{h,r-1\}.$ We only need $\F_{q_0}$-linear independence of any $m+1$ columns of 
	\begin{align*}
		H=\begin{bmatrix}
			1 & 1 & \cdots & 1\\
			\beta_1 & \beta_2 &\cdots &\beta_r
		\end{bmatrix}.
	\end{align*}

	This follows from the fact that the code $C_0$ has minimum distance at least $m+2$, and therefore any $m+1$ columns of the parity check matrix $H$ must be linearly independent.
\end{proof}

To get the asymptotic field size bound, we instantiate the code $C_0$ with BCH codes.
\begin{proposition}
 There exist $[r,r-s,d]_{\F_{q_0}}$ BCH code with $$s=1+\big((d-2)-\lfloor(d-2)/q_0\rfloor\big)\ceil{\log_{q_0}r}.$$
\end{proposition}
\begin{proof}
   Let $\ell=\ceil{\log_{q_0}r}$ so that $q_0^\ell \ge r.$ Choose distinct $\theta_1,\theta_2,\dots,\theta_r\in \F_{q_0^\ell}$. The parity check matrix of the BCH code is given by:
	$$H_{in}=
	\begin{bmatrix}
	    1& 1 & \dots & 1\\
		\theta_1 & \theta_2 &\dots & \theta_r\\
		\vdots&\vdots&\ddots&\vdots\\
		\theta_1^{q_0-1}& \theta_2^{q_0-1} &\dots & \theta_r^{q_0-1}\\
		\theta_1^{q_0+1}& \theta_2^{q_0+1} &\dots & \theta_r^{q_0+1}\\
		\vdots&\vdots&\ddots&\vdots\\
		\theta_1^{d-2}& \theta_2^{d-2} &\dots & \theta_r^{d-2}
	\end{bmatrix},
	$$
	where we removed powers which are multiples of $q_0$. Each row of $H$ other than the first row of 1's should be thought of as $\ell$ rows over the base field $\F_{q_0}.$ Therefore the codimension of the code is $s\le 1+\ell((d-2)-\floor{(d-2)/q_0}).$ Finally, the distance of the code is at least $d$. This is because to argue about $\F_{q_0}$ linear independence of any $d-1$ columns, we can add back the rows whose powers are multiples of $q_0$ to $H$ which is a Vandermonde matrix over $\F_{q_0^\ell}$.
\end{proof}
Therefore we can choose $s=1+\big(m-\lfloor m/q_0\rfloor\big)\ceil{\log_{q_0}r}$ where $m=\min\{h,r-1\}$. Therefore we get a field size of $$q=q_0^{s-1} \le (O\left(n\right))^{m-\floor{m/q_0}}$$.

\section*{Acknowledgment}
We thank Sergey Yekhanin for several illuminating discussions about MR-LRCs and Umberto Mart{\'{\i}}nez{-}Pe{\~{n}}as for helpful comments on an earlier version of this paper.
\bibliographystyle{alpha}
\bibliography{references}

\newcommand{\etalchar}[1]{$^{#1}$}
\begin{thebibliography}{GHK{\etalchar{+}}17}

\bibitem[Ber15]{berlekamp2015algebraic}
Elwyn~R Berlekamp.
\newblock {\em Algebraic coding theory (revised edition)}.
\newblock World Scientific, 2015.

\bibitem[BHH13]{BHH}
Mario Blaum, James~Lee Hafner, and Steven Hetzler.
\newblock Partial-{MDS} codes and their application to {RAID} type of
  architectures.
\newblock {\em IEEE Transactions on Information Theory}, 59(7):4510--4519,
  2013.

\bibitem[Bla13]{Blaum}
Mario Blaum.
\newblock Construction of {PMDS} and {SD} codes extending {RAID} 5.
\newblock Arxiv 1305.0032, 2013.

\bibitem[BPSY16]{BPSY}
Mario Blaum, James Plank, Moshe Schwartz, and Eitan Yaakobi.
\newblock Construction of partial {MDS} and sector-disk codes with two global
  parity symbols.
\newblock {\em IEEE Transactions on Information Theory}, 62(5):2673--2681,
  2016.

\bibitem[BSC{\etalchar{+}}12]{bostan2012power}
Alin Bostan, Bruno Salvy, Muhammad~FI Chowdhury, {\'E}ric Schost, and Romain
  Lebreton.
\newblock Power series solutions of singular (q)-differential equations.
\newblock In {\em Proceedings of the 37th International Symposium on Symbolic
  and Algebraic Computation}, pages 107--114, 2012.

\bibitem[BU14]{boucher2014linear}
Delphine Boucher and Felix Ulmer.
\newblock Linear codes using skew polynomials with automorphisms and
  derivations.
\newblock {\em Designs, codes and cryptography}, 70(3):405--431, 2014.

\bibitem[CHL07]{CHL}
Minghua Chen, Cheng Huang, and Jin Li.
\newblock On maximally recoverable property for multi-protection group codes.
\newblock In {\em IEEE International Symposium on Information Theory (ISIT)},
  pages 486--490, 2007.

\bibitem[CK17]{CK}
Gokhan Calis and Ozan Koyluoglu.
\newblock A general construction fo {PMDS} codes.
\newblock {\em IEEE Communications Letters}, 21(3):452--455, 2017.

\bibitem[CMST21]{cai2020construction}
Han Cai, Ying Miao, Moshe Schwartz, and Xiaohu Tang.
\newblock A construction of maximally recoverable codes with order-optimal
  field size.
\newblock {\em IEEE Transactions on Information Theory}, 68(1):204--212, 2021.

\bibitem[FG15]{FG15}
Michael~A. Forbes and Venkatesan Guruswami.
\newblock Dimension expanders via rank condensers.
\newblock In {\em Proceedings of the 19th International Workshop on
  Randomization and Computation (RANDOM)}, pages 800--814, 2015.

\bibitem[FS12]{FS12}
Michael~A. Forbes and Amir Shpilka.
\newblock On identity testing of tensors, low-rank recovery and compressed
  sensing.
\newblock In {\em Proceedings of the 44th ACM Symposium on Theory of
  Computing}, pages 163--172. {ACM}, 2012.

\bibitem[FSS14]{FSS14}
Michael~A. Forbes, Ramprasad Saptharishi, and Amir Shpilka.
\newblock Hitting sets for multilinear read-once algebraic branching programs,
  in any order.
\newblock In {\em Proceedings of the ACM Symposium on Theory of Computing},
  pages 867--875, 2014.

\bibitem[GGY20]{gopi2020maximally}
Sivakanth Gopi, Venkatesan Guruswami, and Sergey Yekhanin.
\newblock Maximally recoverable {LRC}s: {A} field size lower bound and
  constructions for few heavy parities.
\newblock {\em {IEEE} Trans. Inf. Theory}, 66(10):6066--6083, 2020.

\bibitem[GHJY14]{GHJY}
Parikshit Gopalan, Cheng Huang, Bob Jenkins, and Sergey Yekhanin.
\newblock Explicit maximally recoverable codes with locality.
\newblock {\em IEEE Transactions on Information Theory}, 60(9):5245--5256,
  2014.

\bibitem[GHK{\etalchar{+}}17]{GHKSWY}
Parikshit Gopalan, Guangda Hu, Swastik Kopparty, Shubhangi Saraf, Carol Wang,
  and Sergey Yekhanin.
\newblock Maximally recoverable codes for grid-like topologies.
\newblock In {\em 28th Annual Symposium on Discrete Algorithms (SODA)}, pages
  2092--2108, 2017.

\bibitem[GHSY12]{GHSY}
Parikshit Gopalan, Cheng Huang, Huseyin Simitci, and Sergey Yekhanin.
\newblock On the locality of codeword symbols.
\newblock {\em IEEE Transactions on Information Theory}, 58(11):6925 --6934,
  2012.

\bibitem[GJX20]{GLX-ff}
Venkatesan Guruswami, Lingfei Jin, and Chaoping Xing.
\newblock Constructions of maximally recoverable local reconstruction codes via
  function fields.
\newblock {\em {IEEE} Trans. Inf. Theory}, 66(10):6133--6143, 2020.

\bibitem[GK16]{GK-combinatorica}
Venkatesan Guruswami and Swastik Kopparty.
\newblock Explicit subspace designs.
\newblock {\em Combinatorica}, 36(2):161--185, 2016.

\bibitem[GRX18]{guruswami2018lossless}
Venkatesan Guruswami, Nicolas Resch, and Chaoping Xing.
\newblock Lossless dimension expanders via linearized polynomials and subspace
  designs.
\newblock In {\em 33rd Computational Complexity Conference (CCC 2018)}. Schloss
  Dagstuhl-Leibniz-Zentrum fuer Informatik, 2018.

\bibitem[Gur11]{Gur-ccc11}
Venkatesan Guruswami.
\newblock Linear-algebraic list decoding of folded {R}eed-{S}olomon codes.
\newblock In {\em Proceedings of the 26th IEEE Conference on Computational
  Complexity}, pages 77--85, 2011.

\bibitem[GW11]{GW-derivative-codes}
Venkatesan Guruswami and Carol Wang.
\newblock Optimal rate list decoding via derivative codes.
\newblock In {\em Proceedings of APPROX/RANDOM 2011}, pages 593--604, August
  2011.

\bibitem[GW13]{GW13}
Venkatesan Guruswami and Carol Wang.
\newblock Linear-algebraic list decoding for variants of {R}eed-{S}olomon
  codes.
\newblock {\em IEEE Transactions on Information Theory}, 59(6):3257--3268,
  2013.

\bibitem[GXY18]{GXY-tams}
Venkatesan Guruswami, Chaoping Xing, and Chen Yuan.
\newblock Constructions of subspace designs via algebraic function fields.
\newblock {\em Trans. Amer. Math. Soc.}, 370:8757--8775, 2018.

\bibitem[GYBS17]{GYBS}
Ryan Gabrys, Eitan Yaakobi, Mario Blaum, and Paul Siegel.
\newblock Construction of partial {MDS} codes over small finite fields.
\newblock In {\em 2017 IEEE International Symposium on Information Theory
  (ISIT)}, pages 1--5, 2017.

\bibitem[HCL07]{HCL}
Cheng Huang, Minghua Chen, and Jin Li.
\newblock Pyramid codes: flexible schemes to trade space for access efficiency
  in reliable data storage systems.
\newblock In {\em 6th IEEE International Symposium on Network Computing and
  Applications (NCA 2007)}, pages 79--86, 2007.

\bibitem[HSX{\etalchar{+}}12]{HuangSX}
Cheng Huang, Huseyin Simitci, Yikang Xu, Aaron Ogus, Brad Calder, Parikshit
  Gopalan, Jin Li, and Sergey Yekhanin.
\newblock Erasure coding in {W}indows {A}zure {S}torage.
\newblock In {\em USENIX Annual Technical Conference (ATC)}, pages 15--26,
  2012.

\bibitem[HY16]{HY}
Guangda Hu and Sergey Yekhanin.
\newblock New constructions of {SD} and {MR} codes over small finite fields.
\newblock In {\em 2016 IEEE International Symposium on Information Theory
  (ISIT)}, pages 1591--1595, 2016.

\bibitem[Lam85]{Lam85}
Tsit-Yuen Lam.
\newblock {\em A general theory of Vandermonde matrices}.
\newblock Center for Pure and Applied Mathematics, University of California,
  Berkeley, 1985.

\bibitem[LL88]{LamL88}
Tsit-Yuen Lam and Andr{\'e} Leroy.
\newblock Vandermonde and wronskian matrices over division rings.
\newblock {\em Journal of Algebra}, 119(2):308--336, 1988.

\bibitem[Mar18]{Martinez18}
Umberto Mart{\'{\i}}nez{-}Pe{\~{n}}as.
\newblock Skew and linearized reed--solomon codes and maximum sum rank distance
  codes over any division ring.
\newblock {\em Journal of Algebra}, 504:587--612, 2018.

\bibitem[Mar20]{UMP-2020}
Umberto Mart{\'{\i}}nez{-}Pe{\~{n}}as.
\newblock A general family of {MSRD} codes and {PMDS} codes with smaller field
  sizes from extended {M}oore matrices.
\newblock {\em CoRR}, abs/2011.14109, 2020.

\bibitem[MK19]{MK19}
Umberto Mart{\'{\i}}nez{-}Pe{\~{n}}as and Frank~R. Kschischang.
\newblock Universal and dynamic locally repairable codes with maximal
  recoverability via sum-rank codes.
\newblock {\em {IEEE} Trans. Inf. Theory}, 65(12):7790--7805, 2019.

\bibitem[MPK19]{MartinezK19_NetworkCoding}
Umberto Mart{\'\i}nez-Pe{\~n}as and Frank~R Kschischang.
\newblock Reliable and secure multishot network coding using linearized
  reed-solomon codes.
\newblock {\em IEEE Transactions on Information Theory}, 2019.

\bibitem[MV13]{mahdavifar2013algebraic}
Hessam Mahdavifar and Alexander Vardy.
\newblock Algebraic list-decoding of subspace codes.
\newblock {\em IEEE Transactions on Information Theory}, 59(12):7814--7828,
  2013.

\bibitem[NUF10]{NU10}
Roberto~W N{\'o}brega and Bartolomeu~F Uch{\^o}a-Filho.
\newblock Multishot codes for network coding using rank-metric codes.
\newblock In {\em 2010 Third IEEE International Workshop on Wireless Network
  Coding}, pages 1--6. IEEE, 2010.

\bibitem[Ore33]{ore1933}
Oystein Ore.
\newblock Theory of non-commutative polynomials.
\newblock {\em Annals of mathematics}, pages 480--508, 1933.

\bibitem[PD14]{Dimakis_0}
Dimitris Papailiopoulos and Alexandros Dimakis.
\newblock Locally repairable codes.
\newblock {\em IEEE Transactions on Information Theory}, 60(10):5843--5855,
  2014.

\bibitem[SAP{\etalchar{+}}13]{XOR_ELE}
Maheswaran Sathiamoorthy, Megasthenis Asteris, Dimitris~S. Papailiopoulos,
  Alexandros~G. Dimakis, Ramkumar Vadali, Scott Chen, and Dhruba Borthakur.
\newblock {XOR}ing elephants: novel erasure codes for big data.
\newblock In {\em Proceedings of VLDB Endowment (PVLDB)}, pages 325--336, 2013.

\bibitem[TPD16]{TPD}
Itzhak Tamo, Dimitris Papailiopoulos, and Alexandros~G. Dimakis.
\newblock Optimal locally repairable codes and connections to matroid theory.
\newblock {\em IEEE Transactions on Information Theory}, 62:6661--6671, 2016.

\end{thebibliography}

\appendix

\section{Examples of Skew Polynomial Rings}
\label{sec:examples_skewpolyring}
% \gnote{TODO: Can we include differential equations as an example? Also include discussion on possible skew polynomial rings over finite fields.}
In Section~\ref{sec:prelim}, we discussed a few examples of skew polynomial rings such as when the derivation is the zero map, i.e., $\delta(a)=0$ for all $a\in \K$.  In this case, the skew ring is denoted by $\K[t;\sigma]$ and is said to be of endomorphism type. Here we give a few more interesting examples.
\begin{example}[Skew Polynomial Rings]
	\label{example:skewpolynomialring}
	\begin{enumerate}
		%\item The simplest example of a skew polynomial ring is when $\sigma$ is the identity map and $\delta$ is the zero map. In this case, skew polynomials coincide with the usual notion of polynomials.
		%\item The simplest derivation is the zero map i.e. $\delta(a)=0$ for all $a\in \K$.  In this case, the skew ring is denoted by $\K[t;\sigma]$ and is said to be of endomorphism type.% Skew polynomials are interesting even in this case, and in fact the constructions in this paper only use skew polynomials with $\delta\equiv 0$. So the reader can imagine that the derivation is the zero map on a first reading. We include the general case in the hope that skew polynomial rings with non-zero derivations will find applications in future.
		%\item Let $\K=\F(x)$ be the field of rational functions over $\F$ and let $\gamma\in \F^*$. Then $\sigma(f(x))=f(\gamma x)$ is an endomorphism and $\delta(f(x))=f(\gamma x)-f(x)$ is a $\sigma$-derivation.
		\item Let $\K$ be any field and let $\sigma:\K\to\K$ be an endomorphism. Then for any $\lambda\in \K,$ $\delta(a)=\lambda(\sigma(a)-a)$ is a $\sigma$-derivation.\footnote{If $\K$ is a division ring, then $\delta(a)=\sigma(a)\lambda - \lambda a$ is a $\sigma$-derivation.} These are called inner-derivations and the skew polynomial ring defined using such a derivation is isomorphic to the skew polynomial ring over $\K$ with the same $\sigma$ and $\delta=0.$\footnote{The isomorphism is $\phi: \K[t;\sigma,\delta]\to \K[\tilde t; \sigma]$ defined as $\phi(t)=\tilde t - \lambda$ and $\phi|_\K \equiv \Id$.} The concept of $q$-derivatives~\cite{bostan2012power} is a special case of this for $\K=\F(x)$. For some fixed $q\in \F\setminus \{1\}$, the $q$-derivative $f\in \F(x)$ is defined as $(f(qx)-f(x))/(qx-x).$ This is a derivation w.r.t. the endomorphism $\sigma: f(x) \to f(qx).$
		\item Let $\K=\F(x)$ and $\sigma$ be the identity map. Then $\delta(f(x))$ defined as the formal derivative of $f(x)$ is a $\sigma$-derivation. This can be extended to rational functions in a consistent way using power series.  When $\sigma$ is the identity map, the skew ring is denoted by $\K[t;\delta]$ and is said to be of derivation type.
		\item Let $\K$ be the set of smooth real-valued functions over $\R$ and $\sigma$ be the identity map. Then $\delta(f(x))$ defined as the derivative $f'(x)$ is a $\sigma$-derivation. This is an important skew polynomial ring for the study of linear differential equations. For a skew polynomial $g(t)=g_dt^d+\dots+g_1t+g_0 \in \K[t;\delta]$ and a smooth function $f:\R\to \R$, $\conj{0}{f}$ is a root of $g(t)$ iff $f$ satisfies the linear differential equation $$g_d D^df + \dots + g_1 Df+g_0f=0$$ where $D=\frac{d}{dx}$ is the derivative operator. Theorem~\ref{thm:fundamentalthm_roots_skewpolynomials} implies that the set of roots to $g(t)$ forms a vector space of dimension at most $d$ over the centralizer subfield $\K_0=\{f: \conj{0}{f}=0\}=\{f:f'=0\}\cong \R.$ This is consistent with the well-known fact that the space of solutions of a degree $d$ homogeneous linear differential equation has dimension at most $d$.
	\end{enumerate}
\end{example}

The following two propositions classify skew polynomial rings over fields and finite fields.
\begin{proposition}
	\label{prop:skewrings_fields}
	When $\K$ is a field (as opposed to being a division ring), up to isomorphisms, the only possible skew polynomial rings over $\K$ are either of endomorphism type (i.e., $\delta\equiv 0$) or derivation type (i.e., $\sigma\equiv\Id$). 
\end{proposition}
\begin{proof}
	This is because if $\sigma\ne \Id$, then there exists some element $a_0\in \K$ such that $\sigma(a_0)\ne a_0$. Now using commutativity of $\K$, we have $\delta(aa_0)=\delta(a_0a)$ for any $a\in \K$. Expanding both sides, we get that for any $a\in \K$, $\delta(a)=\lambda (\sigma(a)-a)$ where $\lambda = \delta(a_0)/(\sigma(a_0)-a_0)$ is a fixed constant, i.e., $\delta$ is an inner-derivation. As we discussed above, this skew polynomial ring is isomorphic to the skew polynomial ring with $\delta\equiv 0$ and the same endomorphism $\sigma.$
\end{proof}

\begin{proposition}
	\label{prop:skewrings_finitefields}
	When $\K=\F_q$ is a finite field, up to isomorphisms, the only possible skew polynomial rings are of the endomorphism type (i.e., $\delta\equiv 0$). 
\end{proposition}
\begin{proof}
	By Proposition~\ref{prop:skewrings_fields}, we already know that the skew polynomial ring has to be either of endomorphism type or derivation type. So we just have to rule out the derivation type. Suppose there is a skew polynomial ring of derivation type, i.e., $\sigma\equiv \Id$ and $\delta\ne 0$. Suppose $\mathrm{char}(\F_q)=p$. Then by repeatedly applying chain rule for $\delta$, for any $a\in \K$, $$\delta(a^p)=a\delta(a^{p-1})+\delta(a)a^{p-1}=\cdots = p\delta(a)a^{p-1}=0.$$ This is a contradiction.
\end{proof}

\section{Missing Proofs from Section~\ref{sec:prelim}}
\label{sec:missing_proofs}
\begin{lemma}[Lemma~\ref{lem:conjugacy_equivalence}]
	%\label{lem:conjugacy_equivalence}
	\begin{enumerate}
		\item $\conj{(\conj a c)}{d}=\conj{a}{dc}$
		\item Conjugacy is an equivalence relation, i.e., we can partition $\K$ into conjugacy classes where elements in each part are conjugates of each other, but elements in different parts are not conjugates.
	\end{enumerate}
\end{lemma}
\begin{proof}
	(1) follows easily from the definition of conjugation and the using the fact that $\delta(cd)=\sigma(c)\delta(d)+\delta(c)d.$
	\begin{align*}
		\conj{(\conj{a}{c})}{d}&= \sigma(d)\cdot \conj{a}{c} \cdot d^{-1}+\delta(d)d^{-1}\\
		&= \sigma(d)(\sigma(c)ac^{-1}+\delta(c)c^{-1})d^{-1}+\delta(d)d^{-1}\\
		&= \sigma(dc)ac^{-1}d^{-1}+\sigma(d)\delta(c)c^{-1}d^{-1}+\delta(d)d^{-1}\\
		&= \sigma(dc)a(dc)^{-1}+(\sigma(d)\delta(c)+\delta(d)c)c^{-1}d^{-1}\\
		&= \sigma(dc)a(dc)^{-1}+\delta(dc)(dc)^{-1}\\
		&=\conj{a}{dc}.
	\end{align*}
	We now prove (2). Suppose $a$ is a conjugate of $b$, i.e., $a=\conj{b}{x}$ for some $x\in \K^*.$ Then $\conj{a}{x^{-1}}=\conj{(\conj{b}{x})}{x^{-1}}=\conj{b}{x^{-1}x}=b.$ Therefore $b$ is a conjugate of $a.$ Suppose $a$ is a conjugate of $b$, with $a=\conj{b}{x}$, and $c$ is a conjugate of $b$, with $b=\conj{c}{y}$. Then $a=\conj{b}{x}=\conj{(\conj{c}{y})}{x}=\conj{c}{xy}.$ So $a$ is a conjugate of $c.$
\end{proof}

\begin{lemma}[Lemma~\ref{lem:centralizer_subfield}]
\begin{enumerate}
	\item $\K_a$ is a subfield of $\K.$\footnote{When $\K$ is a division ring, $\K_a$ will be a sub-division ring of $\K.$}
	\item If $a,b\in \K$ are conjugates, then $\K_a=\K_b$. \footnote{When $\K$ is a division ring and not a field, we have $\K_{(\conj{a}{x})}=x\K_ax^{-1}$.}
\end{enumerate}
\end{lemma}
\begin{proof}
	(1) Let $x,y\in \K_a\setminus \{0\}$ i.e. $\conj{a}{x}=\conj{a}{y}=a$. Then
	\begin{align*}
		\conj{a}{x+y}(x+y)&=\sigma(x+y)a + \delta(x+y)\\
		&=\sigma(x)a+\sigma(y)a + \delta(x)+\delta(y)\\
		&=\conj{a}{x}x+\conj{a}{y}y\\
		&=ax+ay=a(x+y).
	\end{align*}
	Therefore $\conj{a}{x+y}=a$. Also $\conj{a}{yx}=\conj{(\conj{a}{x})}{y}=a.$ And finally $\conj{a}{x^{-1}}=\conj{(\conj{a}{x})}{x^{-1}}=\conj{a}{x^{-1}x}=a.$

	(2) Suppose $b=\conj{a}{d}$ and let $c\in \K_a.$Then $\conj{b}{c}=\conj{(\conj{a}{d})}{c}=\conj{a}{cd}=\conj{a}{dc}=\conj{(\conj{a}{c})}{d}=\conj{a}{d}=b.$ Therefore $\K_a\subset \K_b$. By symmetry, $\K_b\subset \K_a.$
\end{proof}

\begin{lemma}[Product evaluation rule (Lemma~\ref{lem:product_evaluation})]
	If $g(a)=0$, then $(fg)(a)=0$. If $g(a)\ne 0$ then $$(fg)(a)=f\left(\conj{a}{g(a)}\right)g(a).$$
\end{lemma}
\begin{proof}
	If $g(a)=0$, then $g(t)=b(t)(t-a)$ for some $b(t)\in \K[t;\sigma,\delta]$. Therefore $f(t)g(t)=f(t)b(t)(t-a)$, and so $(fg)(a)=0$. Suppose $g(a)\ne 0.$ Let $g(t)=b(t)(t-a)+g(a)$ and $f(t)=a(t)\left(t-\conj{a}{g(a)}\right)+f\left(\conj{a}{g(a)}\right)$. Then 
	\begin{align*}
		f(t)g(t) &= f(t)\cdot (b(t)(t-a)+g(a))\\
		&= f(t)b(t)(t-a)+f(t)g(a)\\
		&=f(t)b(t)(t-a)+\left(a(t)\left(t-\conj{a}{g(a)}\right)+f\left(\conj{a}{g(a)}\right)\right)g(a)\\
		&=f(t)b(t)(t-a)+a(t)\left(tg(a)-\conj{a}{g(a)}\cdot g(a)\right)+f\left(\conj{a}{g(a)}\right)g(a)\\
		&=f(t)b(t)(t-a)+a(t)\left(\sigma(g(a))t+\delta(g(a))-\sigma(g(a))a-\delta(g(a))\right)+f\left(\conj{a}{g(a)}\right)g(a)\\
		&=f(t)b(t)(t-a)+a(t)\sigma(g(a))(t-a)+f\left(\conj{a}{g(a)}\right)g(a)\\
		&=\left(f(t)b(t)+a(t)\sigma(g(a))\right)(t-a)+f\left(\conj{a}{g(a)}\right)g(a).
	\end{align*}
	Therefore $(fg)(a)=f\left(\conj{a}{g(a)}\right)g(a).$
\end{proof}

\section{Roots of Skew Polynomials}
\label{sec:proof_fundamental_thm}
The most important and useful fact about usual polynomials is that a degree $d$ non-zero polynomial can have at most $d$ roots. It turns out that this statement is false for skew polynomials! A skew polynomial can have many more roots than its degree. But when counted in the right way, we can recover an analogous statement for skew polynomials. In this section, we will prove the ``fundamental theorem" about roots of skew polynomials which shows that a degree $d$ skew polynomial cannot have more than $d$ roots when counted the right way.
We will begin with showing that any non-zero degree $d$ skew polynomial can have at most $d$ roots in distinct conjugacy classes. 

\begin{lemma}
	\label{lem:distinct_conjugate_roots}
	Let $f\in \K[t;\sigma,\delta]$ be a degree $d$ non-zero polynomial. Then $f$ can have at most $d$ roots in distinct conjugacy classes.
\end{lemma}
\begin{proof}
	We will prove it using induction on the degree. For the base case, it is clear that a degree $0$ polynomial which is a non-zero constant cannot have any roots. Suppose $a_0,a_1,\dots,a_d\in \K$ be roots of $f$ in distinct conjugacy classes. Since $f(a_0)=0,$ we can write $f(t)=h(t)(t-a_0)$ where $\deg(h)=d-1.$ By Lemma~\ref{lem:product_evaluation}, $f(a_i)=h(\conj{a_i}{a_i-a_0})(a_i-a_0)$. Therefore $b_i=\conj{a_i}{a_i-a_0}$ for $i\in \set{1,\dots,d}$ are $d$ roots of $h$ and they lie in distinct conjugacy classes because $a_i$ lie in distinct conjugacy classes. Thus by induction $h=0$ and therefore $f=0$ which is a contradiction.
\end{proof}

Now let us try to understand, the roots of a skew polynomial in the same conjugacy class. Let $f\in \K[t;\sigma,\delta]$ be a non-zero polynomial and fix some $a\in \K$ and let $\K_a$ be the centralizer of $a$ (which is a subfield of $\K$). Define $V_f(a)=\{y\in \K^*: f(\conj{a}{y})=0\}\cup \set{0}$. Lemma~\ref{lem:conjugate_roots_vectorspace} shows that $V_f(a)$ is a vector space over $\K_a.$
The next lemma shows that the dimension of $V_f(a)$ can be at most $\deg(f).$
\begin{lemma}
	\label{lem:conjugate_roots_dim}
	Let $f\in \K[t;\sigma,\delta]$ be a degree $d$ non-zero polynomial and fix some $a\in \K$ and let $\F=\K_a$ be the centralizer subfield of $a$. Define $V_f(a)=\{y\in \K^*: f(\conj{a}{y})=0\}\cup \set{0}$. Then $V_f(a)$ is a vector space over $\F$ of dimension at most $d.$
\end{lemma}
\begin{proof}
	We will use induction on the degree. For the base case, it is clear that for a degree $0$ polynomial, which is a non-zero constant, $\dim_\F(V_f(a))=0$. Suppose for contradiction that there exists $y_0,y_1,\dots,y_d\in V_f(a)$ which are linearly independent over $\F$. WLOG, we can assume that $y_0=1$ (by redefining $a$ to be equal to $\conj{a}{y_0}$). Since $f(a)=0,$ we can write $f(t)=h(t)(t-a)$ where $\deg(h)=d-1.$ By Lemma~\ref{lem:product_evaluation}, $f(\conj{a}{y_i})=h(\conj{a}{y_i(\conj{a}{y_i}-a)})(\conj{a}{y_i}-a)$. Since $y_0=1$ and $y_i$ is linearly independent from $y_0$ over $\F$, $y_i\notin \F$. Therefore $\conj{a}{y_i}-a\ne 0$, and so $b_i=\conj{a}{y_i(\conj{a}{y_i}-a)}$ for $i\in \set{1,\dots,d}$ are $d$ roots of $h$. If we show that $y_i(\conj{a}{y_i}-a)$ for $i\in \set{1,\dots,d}$ are linearly independent over $\F$, then we are done by induction.

	Suppose they are not independent. Then there exists $c_1,\dots,c_d \in \F$ s.t. $\sum_{i=1}^d c_iy_i(\conj{a}{y_i}-a)=0$. Therefore,
	\begin{align*}
		a\sum_{i=1}^d c_iy_i &= \sum_{i=1}^d c_iy_i\cdot \conj{a}{y_i}\\
							 &= \sum_{i=1}^d c_iy_i \cdot \conj{a}{c_iy_i} \tag{$c_i\in \F=\K_a$}\\
							 &= \left(\sum_{i=1}^d c_iy_i\right)\conj{a}{\left(\sum_{i=1}^d c_iy_i\right)} \tag{$\conj{a}{x+y}(x+y)=\conj{a}{x}x+\conj{a}{y}y$ for all $x,y\in \K^*$}
	\end{align*}
	Since $y_1,\dots,y_d$ are independent over $\F$, $\sum_{i=1}^d c_iy_i\ne 0.$ Therefore $\conj{a}{\left(\sum_{i=1}^d c_iy_i\right)}=a$ i.e. $\sum_{i=1}^d c_iy_i \in \K_a=\F$. But this contradicts the fact that $\set{y_0=1,y_1,\dots,y_d}$ are linearly independent over $\F.$
\end{proof}

We will now prove the ``fundamental theorem" about roots of skew polynomials. It immediately implies Lemma~\ref{lem:distinct_conjugate_roots} and Lemma~\ref{lem:conjugate_roots_dim} as corollaries. But we have proved them before, just to convey some intuition. 
\begin{theorem}[Theorem~\ref{thm:fundamentalthm_roots_skewpolynomials}]
	%\label{thm:fundamentalthm_roots_skewpolynomials}
	Let $f\in \K[t;\sigma,\delta]$ be a degree $d$ non-zero polynomial. Let $A$ be the set of roots of $f$ in $\K$ and let $A=\cup_i A_i$ be a partition of $A$ into conjugacy classes. Fix some representatives $a_i\in A_i$. Let $V_i=\{y: \conj{a_i}{y}\in A_i\}\cup \{0\}$ which is a linear subspace over $\F_i=\K_{a_i}$ by Lemma~\ref{lem:conjugate_roots_vectorspace}. Then $$\sum_i \dim_{\F_i}(V_i)\le d.$$
\end{theorem}
\begin{proof}

	We will use induction on the degree. For the base case, it is clear that for a degree $0$ polynomial, which is a non-zero constant, $\dim_{\F_i}(V_i)=0$ for every $i$. We will now show the induction step.

	For each $i$, let $d_i=\dim_{\F_i}(V_i).$ Fix some basis $y(i,1),y(i,2),\dots,y(i,d_i)\in \K^*$ which span $V_i$ with coefficients in $\F_i=\K_{a_i}$. WLOG, we can assume that $y(i,1)=1$ for every $i$, by reassigning $a_i=\conj{a_i}{y(i,1)}$.

	Fix some conjugacy class $i^*$ s.t. $d_{i*}\ge 1$. Since $f(a_{i^*})=0,$ we can write $f(t)=h(t)(t-a_{i^*})$ where $\deg(h)=d-1.$ Now let $A_i'$ be the roots of $h$ in conjugacy class $i$ and $V_i'=\{y:\conj{a_i}{y}\in A_i'\}\cup \{0\}$. We claim that $\dim_{\F_i}(V_i')\ge \dim_{\F_i}(V_i)$ for every $i\ne i^*$ and $\dim_{\F_{i^*}}(V_{i^*}')\ge \dim_{\F_{i^*}}(V_{i^*})-1$. By induction $\sum_i \dim_{\F_i}(V_i') \le d-1$. Therefore we have $\sum_i \dim_{\F_i}(V_i) \le d.$ We will now prove the claim in two parts.

	\begin{claim}
		$\dim_{\F_i}(V_i')\ge \dim_{\F_i}(V_i)$ for every $i\ne i^*$.
	\end{claim}
	\begin{proof}
		Fix some conjugacy class $i\ne i^*$. By Lemma~\ref{lem:product_evaluation}, $$f\left(\conj{a_i}{y(i,j)}\right)=h\left(\conj{a_i}{y(i,j)\left(\conj{a_i}{y(i,j)}-a_{i^*}\right)}\right)\left(\conj{a_i}{y(i,j)}-a_{i^*}\right).$$ Since $a_i,a_{i^*}$ are in different conjugacy classes, $\conj{a_i}{y(i,j)}-a_{i^*}\ne 0$. So $b_j=\conj{a_i}{y(i,j)(\conj{a_i}{y(i,j)}-a_{i^*})}$ for $j\in \set{1,\dots,d_i}$ are $d_i$ roots of $h$ in the $i^{th}$ conjugacy class $A_i'$. If we show that $y(i,j)(\conj{a_i}{y(i,j)}-a_{i^*})$ for $j\in \set{1,\dots,d_i}$ are linearly independent over $\F_i$, then this proves the claim.

		Suppose they are not independent. Then there exists $c_1,\dots,c_{d_i} \in \F_i$ s.t. $\sum_{j=1}^{d_i} c_jy(i,j)(\conj{a_i}{y(i,j)}-a_{i^*})=0$. Therefore,
		\begin{align*}
			a_{i^*}\sum_{j=1}^{d_i}c_jy(i,j) &= \sum_{j=1}^{d_i}c_jy(i,j)\cdot \conj{a_i}{y(i,j)}\\
								 &= \sum_{j=1}^{d_i}c_jy(i,j) \cdot \conj{a_i}{c_jy(i,j)} \tag{$c_j\in \F_i=\K_{a_i}$}\\
								 &= \left(\sum_{i=1}^{d_i}c_jy(i,j)\right)\conj{a_i}{\left(\sum_{j=1}^{d_i} c_jy(i,j)\right)} \tag{$\conj{a}{x+y}(x+y)=\conj{a}{x}x+\conj{a}{y}y$ for all $x,y\in \K^*$}
		\end{align*}
		Since $y(i,1),\dots,y(i,d_i)$ are independent over $\F_i$, $\sum_{j=1}^{d_i} c_jy(i,j)\ne 0.$ Therefore $\conj{a_i}{\left(\sum_{j=1}^{d_i} c_jy(i,j)\right)}=a_{i^*}$.
		This is a contradiction because $a_i,a_{i^*}$ are in different conjugate classes.		 	
	\end{proof}

	\begin{claim}
		$\dim_{\F_{i^*}}(V_{i^*}')\ge \dim_{\F_{i^*}}(V_{i^*})-1$.
	\end{claim}
	\begin{proof}
		 The proof is exactly similar to that of the previous claim, up until the last.
		 Let $j\in \set{2,3,\dots,d_{i^*}}$. By Lemma~\ref{lem:product_evaluation}, $$f\left(\conj{a_{i^*}}{y(i^*,j)}\right)=h\left(\conj{a_{i^*}}{y(i^*,j)\left(\conj{a_{i^*}}{y(i^*,j)}-a_{i^*}\right)}\right)\left(\conj{a_{i^*}}{y(i^*,j)}-a_{i^*}\right).$$ Since $y(i^*,1)=1$ and $y(i^*,j)$ are linearly independent over $\F_{i^*}$, $y(i^*,j)\notin \F_{i^*}$. Therefore $\conj{a_{i^*}}{y(i^*,j)}-a_{i^*}\ne 0$. So $b_j=\conj{a_{i^*}}{y(i^*,j)(\conj{a_{i^*}}{y(i^*,j)}-a_{i^*})}$ for $j\in \set{2,\dots,d_{i^*}}$ are $d_{i^*}-1$ roots of $h$ in the ${i^*}^{th}$ conjugacy class $A_{i^*}'$. If we show that $y(i^*,j)(\conj{a_{i^*}}{y(i^*,j)}-a_{i^*})$ for $j\in \set{2,\dots,d_{i^*}}$ are linearly independent over $\F_{i^*}$, then this proves the claim.

		 Suppose they are not independent. Then there exists $c_2,\dots,c_{d_{i^*}} \in \F_{i^*}$ s.t.
                 \[ \sum_{j=2}^{d_{i^*}} c_jy(i^*,j)(\conj{a_{i^*}}{y(i^*,j)}-a_{i^*})=0 \ . \]
                 Therefore,
		\begin{align*}
			a_{i^*}\sum_{j=2}^{d_{i^*}}c_jy(i^*,j) &= \sum_{j=2}^{d_{i^*}}c_jy(i^*,j)\cdot \conj{a_{i^*}}{y(i^*,j)}\\
								 &= \sum_{j=2}^{d_{i^*}}c_jy(i^*,j) \cdot \conj{a_{i^*}}{c_jy(i^*,j)} \tag{$c_j\in \F_{i^*}=K_{a_{i^*}}$}\\
								 &= \left(\sum_{j=2}^{d_{i^*}}c_jy(i^*,j)\right)\conj{a_{i^*}}{\left(\sum_{j=2}^{d_{i^*}} c_jy(i^*,j)\right)} \tag{$\conj{a}{x+y}(x+y)=\conj{a}{x}x+\conj{a}{y}y$ for all $x,y\in \K^*$}
		\end{align*}
		Since $y(i^*,1),\dots,y(i^*,d_{i^*})$ are independent over $\F_{i^*}$, $\sum_{j=2}^{d_{i^*}} c_jy(i^*,j)\ne 0.$ Therefore
                \[ \conj{a_{i^*}}{\left(\sum_{j=2}^{d_{i^*}} c_jy(i^*,j)\right)}=a_{i^*} \ , \]
                and thus $\sum_{j=2}^{d_{i^*}} c_jy(i^*,j) \in K_{a_{i^*}}=\F_{i^*}$. But this contradicts the fact that
                \[ \set{y(i^*,1)=1,y(i^*,2),\dots,y(i^*,d_{i^*})} \] are linearly independent over $\F_{i^*}.$		
	\end{proof}
	
	The above two claims finish the proof of Theorem~\ref{thm:fundamentalthm_roots_skewpolynomials}.	
\end{proof}

%!TEX root=./ImprovedMRLRCs.tex

\newcommand{\tbeta}{\widetilde{\beta}}
\section{Constructions of MR LRCs where global parities are outside local groups}

Sometimes, it is better to keep the global parities outside the local groups, i.e., the global/heavy parities do not participate in any local groups. For a given length of the code, this reduces the size of local groups and therefore improves the reconstruction performance (at the cost of slight decrease in durability). Figure~\ref{Fig:LRC_global_outside} shows such an MR LRC. The encoding is done by partitioning the $k$ data symbols into $g$ local groups of size $r-a$ each and adding `$a$' local parities per local group. There are a total of $g$ local groups. Further an additional $h$ global parity checks are added which are placed outside the local groups. The length of the code is therefore $n=k+h+a \cdot \frac{k}{r-a}.$
\begin{figure}[h]
\includegraphics[scale=0.45,trim={0cm 0cm 0cm 0cm},clip]{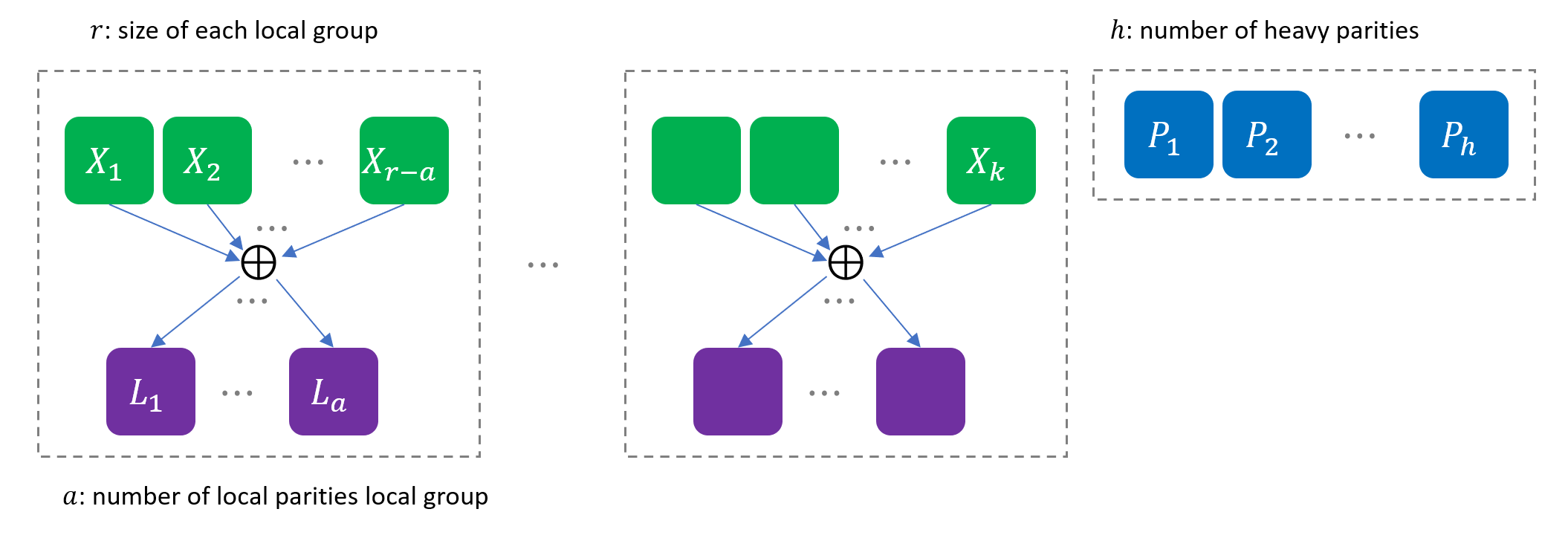}
\caption{An LRC with $k$ data symbols, $h$ global/heavy parities and `$a$' local parities per local group. The global parities are outside the local groups. The length of the code $n=k+h+a \cdot \frac{k}{r-a}$.}
\label{Fig:LRC_global_outside}
\end{figure}
The parity check matrix of an $(n,r,h,a,q)$-MR LRC where the global parities are outside local groups is of the following form:

\begin{equation}\label{fig:MRtopology_global_outside}
H=
\left[
\begin{array}{c|c|c|c|c}
A_1 & 0 & \cdots & 0 & 0\\
\hline
0 &A_2 & \cdots & 0 & 0\\
\hline
\vdots & \vdots & \ddots & \vdots & \vdots \\
\hline
0 & 0 & \cdots & A_g  & 0\\
\hline
B_1 & B_2 & \cdots & B_g & B_{global}\\
\end{array}
\right].
\end{equation}
Here $g=n/r$ is the number of local groups. $A_1,A_2,\dots,A_g$ are $a\times r$ matrices over $\F_q$ which correspond to the local parity checks that each local group satisfies. $B_1,B_2,\dots,B_g$ are $h\times r$ matrices over $\F_q$ and $B_{global}$ is a $h\times h$ matrix; together they represent the $h$ global parity checks that the codewords should satisfy.

The set of correctable erasure patterns correctable by such an MR LRC are exactly those obtained by erasing `$a$' symbols per local group and $h$ additional symbols arbitrarily. Our constructions can be easily modified to obtain the constructions in this setting as well. For simplicity, we will only state the theorems for the case when $h\le r-a$, since this is the regime that is commonly used in practice. The constructions can be easily modified to also work when $h>r-a$.

\begin{theorem}
	\label{thm:main_construction_globaloutside}
	Suppose $h\le r-a$. Let $q_0$ be any prime power such that one of the following is true:
	\begin{enumerate}
		\item $q_0\ge \max\{g+2,r-1\}$ or
		\item $q_0\ge \max\{g+1,r+\ceil{(h-1)/g}-1\}$
	\end{enumerate}
	 where $g=n/r$ is the number of local groups. Then there exists an explicit $(n,r,h,a,q)$-MR LRC with $q=q_0^h$.
\end{theorem}

MR LRCs used in practice typically have only one local parity per local group, i.e., $a=1$~\cite{HuangSX}. We can further improve the construction from Theorem~\ref{thm:main_construction_globaloutside} in this regime.

\begin{theorem}
	\label{thm:construction_a1_globaloutside}
	Suppose $h\le r-a$ and the number of local parities $a=1$. Choose a prime power $q_0$ and a positive integer $n_0$ such that one of the following is true:
	\begin{enumerate}
	 	\item  $q_0 \ge g+2$ and $n_0=r$ or
	 	\item $q_0 \ge g+1$ and $n_0=r+\ceil{\frac{h-1}{g}}$.
	 \end{enumerate}  Suppose there exists an $[n_0,n_0-c,d]_{\F_{q_0}}$ linear code $C_0$ where $c$ is its codimension and minimum distance $d \ge h+2$. Further we need the dual code $C_0^\perp$ to have a codeword of weight exactly $r$.\footnote{This is equivalent to $C_0$ having a parity check matrix containing a row with exactly $r$ non-zero entries.} Then there exists an explicit $(n,r,h,a=1,q)$-MR LRC with field size $$q=q_0^{c-1}$$ where the global parities are outside the local groups.
\end{theorem}

\subsection{Construction: Proof of Theorem~\ref{thm:main_construction_globaloutside}}

Let us recall the parity check matrix of an $(n,r,h,a,q)$-LRC where the global parities are outside local groups is of the form given in (\ref{fig:MRtopology_global_outside}). By Proposition~\ref{prop:MR_LRC_paritycheck}, $C$ is an MR LRC iff (1) any `$a$' columns of each matrix $A_i$ are linearly independent and (2) any submatrix of $H$ formed by selecting $a$ columns in each local group and any $h$ additional columns is full rank.

\noindent \textbf{Case 1: } Let ${q_0}\ge \max\{g+2,r-1\}$ be a prime power.\\
Since we have one extra conjugacy class (note that $q_0-1 \ge g+1$), we will use it to define $B_{global}$. Let $H_0$ be an $(a+h)\times r$ MDS matrix over $\F_{q_0}$ which can be constructed using a Reed-Solomon code. Partition $H_0$ as follows:
 $$
 H_0=\left[\begin{array}{c}
 	A_{a\times r}\\
 	\hline
 	\beta_{h \times r}
 \end{array}\right].
 $$
Define $A_1,A_2,\dots,A_g=A$ where $A$ is formed by the first $a$ rows of $H_0$. Let $\beta_1,\beta_2,\dots,\beta_r \in \F_{q_0}^h$ be the columns of $\beta_{h\times r}$, i.e., 
$$
\beta=\begin{bmatrix}
	\beta_1 & \beta_2 &\cdots & \beta_r
\end{bmatrix}.
$$
Let $\tbeta_i = e_i \in F_{q_0}^h$ for $1\le i \le h$ where $e_1,e_2,\dots,e_h$ are coordinate basis vectors.
Note that $\beta_1,\dots,\beta_r$ and $\tbeta_1,\dots,\tbeta_h$ can also be thought of as elements of $\F_{q_0^h}$ by fixing some basis of $\F_{q_0^h}$ as a vector space over $\F_{q_0}.$ 

Define for $1\le \ell \le g,$
\begin{equation}
\label{eq:choice-of-B-globaloutside}
B_\ell=
\begin{bmatrix}
	 	\beta_1& \beta_2 & \dots & \beta_{r}\\
	 	\gamma^{\ell}\beta_1^{{q_0}}& \gamma^{\ell}\beta_2^{{q_0}}& \cdots & \gamma^{\ell}\beta_{r}^{{q_0}}\\
	 	\gamma^{\ell(1+{q_0})}\beta_1^{q_0^2}& \gamma^{\ell(1+{q_0})}\beta_2^{q_0^2}& \cdots & \gamma^{\ell(1+{q_0})}\beta_{r}^{q_0^2}\\
	 	\vdots & \vdots &  &\vdots\\
	 	\gamma^{\ell(1+{q_0}+\dots+q_0^{h-2})}\beta_1^{q_0^{h-1}}& \gamma^{\ell(1+{q_0}+\dots+q_0^{h-2})}\beta_2^{q_0^{h-1}}& \cdots & \gamma^{\ell(1+{q_0}+\dots+q_0^{h-2})}\beta_{r}^{q_0^{h-1}}
	 \end{bmatrix},
\end{equation}
\begin{equation}
\label{eq:choice-of-B-case1-globaloutside}
B_{global}=
\begin{bmatrix}
	 	\tbeta_1& \tbeta_2 & \dots & \tbeta_{h}\\
	 	\gamma^{(g+1)}\tbeta_1^{{q_0}}& \gamma^{(g+1)}\tbeta_2^{{q_0}}& \cdots & \gamma^{(g+1)}\tbeta_{h}^{{q_0}}\\
	 	\gamma^{(g+1)(1+{q_0})}\tbeta_1^{q_0^2}& \gamma^{(g+1)(1+{q_0})}\tbeta_2^{q_0^2}& \cdots & \gamma^{(g+1)(1+{q_0})}\tbeta_{h}^{q_0^2}\\
	 	\vdots & \vdots &  &\vdots\\
	 	\gamma^{(g+1)(1+{q_0}+\dots+q_0^{h-2})}\tbeta_1^{q_0^{h-1}}& \gamma^{(g+1)(1+{q_0}+\dots+q_0^{h-2})}\tbeta_2^{q_0^{h-1}}& \cdots & \gamma^{(g+1)(1+{q_0}+\dots+q_0^{h-2})}\tbeta_{h}^{q_0^{h-1}}
	 \end{bmatrix}.
\end{equation}

\begin{claim}
	The above construction is an MR LRC over fields of size ${q}=q_0^{h}.$
\end{claim}
\begin{proof}[Proof Sketch]
The proof is very similar to that of Claim~\ref{claim:MRLRC_main}. The only difference is that some of the $h$ additional erasures can happen in the global parities. Since we defined $B_{global}$ so that it belongs to a conjugacy class distinct from those of $B_1,B_2,\dots,B_g$, the proof follows similarly.
\end{proof}

\noindent \textbf{Case 2: } Let ${q_0}\ge \max\{g+1,r+\ceil{\frac{h-1}{g}}\}$ be a prime power.\\
In this case, we don't have an extra (non-zero) conjugacy class to define $B_{global}$. Therefore, we will partition $B_{global}$ into $g$ parts and fold in the parts into the existing $g$ conjugacy classes. Note that we can always include the last column of $B_{global}$ as $[0,0,\dots,1]^T$. Therefore we only need to fold in $h-1$ columns of $B_{global}$ into existing $g$ conjugacy classes. Let $t=\ceil{\frac{h-1}{g}}$ and let $n_0 = r + t$.
Let $H_0$ be an $(a+h)\times n_0$ MDS matrix over $\F_{q_0}$ of the following form
 \begin{equation}
 \label{eqn:H0_def}
 H_0=\left[\begin{array}{c|c}
 	A_{a\times r} & 0\\
 	\hline
 	\beta_{h \times r} & \tbeta_{h \times t}
 \end{array}\right].
 \end{equation}
Note that we can construct an MDS matrix of this form, by first starting with a Reed-Solomon MDS matrix over $\F_{q_0}$ and doing row operations to get this form. Moreover, note that the matrix $A_{a\times r}$ is itself an MDS matrix.

Define $A_1,A_2,\dots,A_g=A$. Let $\beta_1,\beta_2,\dots,\beta_r \in \F_{q_0}^h$ be the columns of $\beta_{h\times r}$, i.e., 
$$
\beta=\begin{bmatrix}
	\beta_1 & \beta_2 &\cdots & \beta_r
\end{bmatrix}.
$$ and define $\tbeta_1,\dots,\tbeta_t\in \F_{q_0}^h$ be the columns of $\tbeta_{h\times t}$.
Note that $\beta_1,\dots,\beta_r$ and $\tbeta_1,\dots,\tbeta_t$ can also be thought of as elements of $\F_{q_0^h}$ by fixing some basis of $\F_{q_0^h}$ as a vector space over $\F_{q_0}.$

Define for $1\le \ell \le g,$ $B_\ell$ as in (\ref{eq:choice-of-B-globaloutside}) and
% \begin{equation}
% %\label{eq:choice-of-B}
% B_\ell=
% \begin{bmatrix}
% 	 	\beta_1& \beta_2 & \dots & \beta_{r}\\
% 	 	\gamma^{\ell}\beta_1^{{q_0}}& \gamma^{\ell}\beta_2^{{q_0}}& \cdots & \gamma^{\ell}\beta_{r}^{{q_0}}\\
% 	 	\gamma^{\ell(1+{q_0})}\beta_1^{q_0^2}& \gamma^{\ell(1+{q_0})}\beta_2^{q_0^2}& \cdots & \gamma^{\ell(1+{q_0})}\beta_{r}^{q_0^2}\\
% 	 	\vdots & \vdots &  &\vdots\\
% 	 	\gamma^{\ell(1+{q_0}+\dots+q_0^{h-2})}\beta_1^{q_0^{h-1}}& \gamma^{\ell(1+{q_0}+\dots+q_0^{h-2})}\beta_2^{q_0^{h-1}}& \cdots & \gamma^{\ell(1+{q_0}+\dots+q_0^{h-2})}\beta_{r}^{q_0^{h-1}}
% 	 \end{bmatrix},
% \end{equation}
\begin{equation}
\label{eq:choice-of-B-case2-globaloutside}
B_{global}^\ell=
\begin{bmatrix}
	 	\tbeta_1& \tbeta_2 & \dots & \tbeta_{t}\\
	 	\gamma^{\ell}\tbeta_1^{{q_0}}& \gamma^{\ell}\tbeta_2^{{q_0}}& \cdots & \gamma^{\ell}\tbeta_{t}^{{q_0}}\\
	 	\gamma^{\ell(1+{q_0})}\tbeta_1^{q_0^2}& \gamma^{\ell(1+{q_0})}\tbeta_2^{q_0^2}& \cdots & \gamma^{\ell(1+{q_0})}\tbeta_{t}^{q_0^2}\\
	 	\vdots & \vdots &  &\vdots\\
	 	\gamma^{\ell(1+{q_0}+\dots+q_0^{h-2})}\tbeta_1^{q_0^{h-1}}& \gamma^{\ell(1+{q_0}+\dots+q_0^{h-2})}\tbeta_2^{q_0^{h-1}}& \cdots & \gamma^{\ell(1+{q_0}+\dots+q_0^{h-2})}\tbeta_{t}^{q_0^{h-1}}
	 \end{bmatrix}.
\end{equation}
Define 
\begin{equation}
\label{eq:choice-of-tBglobal-case2}
	\widetilde{B}_{global}=
\left[\begin{array}{c|c|c|c|c}
	B_{global}^1 & B_{global}^2 & \cdots & B_{global}^g & e_h
\end{array}
\right]
\end{equation}

where $e_h \in F_{q_0}^h$ is the coordinate vector $e_h=[0,0,\dots,0,1]^T$. Note that $\widetilde{B}_{global}$ is an $h\times (gt+1)$ matrix and $gt+1 \ge h$. Finally define $B_{global}$ to be an $h\times h$ matrix formed by arbitrary $h$ columns of $\widetilde{B}_{global}$.

\begin{claim}
	The above construction is an MR LRC over fields of size ${q}=q_0^{h}.$
\end{claim}
\begin{proof}[Proof Sketch]
The proof is very similar to that of Claim~\ref{claim:MRLRC_main}. The only difference is that some of the $h$ additional erasures can happen in the global parities. We will need to crucially use the fact the top right corner of the matrix $H_0$ in (\ref{eqn:H0_def}) used to define $A$'s and $B$'s is zero. Therefore using some `$a$' columns of $A$ to remove the rest of the columns in the upper half of $H_0$, does not affect the $\tbeta$ matrix since the top right corner is already forced to be zero.
\end{proof}

\subsection{Construction: Proof of Theorem~\ref{thm:construction_a1_globaloutside}}

Let us recall the parity check matrix of an $(n,r,h,a,q)$-LRC where the global parities are outside local groups is of the form given in (\ref{fig:MRtopology_global_outside}). By Proposition~\ref{prop:MR_LRC_paritycheck}, $C$ is an MR LRC iff (1) any `$a$' columns of each matrix $A_i$ are linearly independent and (2) any submatrix of $H$ formed by selecting $a$ columns in each local group and any $h$ additional columns is full rank.

\noindent \textbf{Case 1: }  ${q_0}\ge g+2$ and $n_0=r$.\\
Since we have one extra conjugacy class (note that $q_0-1 \ge g+1$), we will use it to define $B_{global}$. Let $C_0$ be an $[n_0,n_0-c,d]_{\F_{q_0}}$ linear code with minimum distance $d\ge h+2$. Let $H_0$ be the parity check matrix of $C_0$ which is a $c \times r$ matrix over $\F_{q_0}$. By the hypothesis that the dual code of $C_0$ has a full weight vector, we can assume that the first row of $H_0$ is all $1$'s vector (scaling the columns if necessary). Partition $H_0$ as follows:
 $$
 H_0=\left[\begin{array}{c c c c}
 	1 & 1 & \cdots & 1\\
 	\hline
 	\beta_1 & \beta_2 & \cdots & \beta_r
 \end{array}\right].
 $$
Define $$A_1,A_2,\dots,A_g=\begin{bmatrix}1&1&\cdots&1\end{bmatrix}.$$ Note that $c\ge h+1$ since any $h+1$ columns of $H_0$ are linearly independent. Therefore we have $h\le c-1$, and so we can define for $1\le i \le h$, 
$$\tbeta_i = e_i$$ where $e_i\in \F_{q_0}^{c-1}$ is the $i^{th}$ coordinate basis vector.

Note that $\beta_1,\dots,\beta_r$ and $\tbeta_1,\dots,\tbeta_h$ can also be thought of as elements of $\F_{q_0^{c-1}}$ by fixing some basis of $\F_{q_0^{c-1}}$ as a vector space over $\F_{q_0}.$

Define for $1\le \ell \le g,$ $B_\ell$ as in (\ref{eq:choice-of-B-globaloutside}) and $B_{global}$ as in (\ref{eq:choice-of-B-case1-globaloutside}).
% \begin{equation}
% %\label{eq:choice-of-B}
% B_\ell=
% \begin{bmatrix}
% 	 	\beta_1& \beta_2 & \dots & \beta_{r}\\
% 	 	\gamma^{\ell}\beta_1^{{q_0}}& \gamma^{\ell}\beta_2^{{q_0}}& \cdots & \gamma^{\ell}\beta_{r}^{{q_0}}\\
% 	 	\gamma^{\ell(1+{q_0})}\beta_1^{q_0^2}& \gamma^{\ell(1+{q_0})}\beta_2^{q_0^2}& \cdots & \gamma^{\ell(1+{q_0})}\beta_{r}^{q_0^2}\\
% 	 	\vdots & \vdots &  &\vdots\\
% 	 	\gamma^{\ell(1+{q_0}+\dots+q_0^{h-2})}\beta_1^{q_0^{h-1}}& \gamma^{\ell(1+{q_0}+\dots+q_0^{h-2})}\beta_2^{q_0^{h-1}}& \cdots & \gamma^{\ell(1+{q_0}+\dots+q_0^{h-2})}\beta_{r}^{q_0^{h-1}}
% 	 \end{bmatrix},
% \end{equation}
% \begin{equation}
% %\label{eq:choice-of-B}
% B_{global}=
% \begin{bmatrix}
% 	 	\tbeta_1& \tbeta_2 & \dots & \tbeta_{h}\\
% 	 	\gamma^{(g+1)}\tbeta_1^{{q_0}}& \gamma^{(g+1)}\tbeta_2^{{q_0}}& \cdots & \gamma^{(g+1)}\tbeta_{h}^{{q_0}}\\
% 	 	\gamma^{(g+1)(1+{q_0})}\tbeta_1^{q_0^2}& \gamma^{(g+1)(1+{q_0})}\tbeta_2^{q_0^2}& \cdots & \gamma^{(g+1)(1+{q_0})}\tbeta_{h}^{q_0^2}\\
% 	 	\vdots & \vdots &  &\vdots\\
% 	 	\gamma^{(g+1)(1+{q_0}+\dots+q_0^{h-2})}\tbeta_1^{q_0^{h-1}}& \gamma^{(g+1)(1+{q_0}+\dots+q_0^{h-2})}\tbeta_2^{q_0^{h-1}}& \cdots & \gamma^{(g+1)(1+{q_0}+\dots+q_0^{h-2})}\tbeta_{h}^{q_0^{h-1}}
% 	 \end{bmatrix}.
% \end{equation}

\begin{claim}
	The above construction is an MR LRC over fields of size ${q}=q_0^{c-1}.$
\end{claim}
\begin{proof}[Proof Sketch]
The proof is very similar to that of Claim~\ref{claim:MRLRC_main}. The only difference is that some of the $h$ additional erasures can happen in the global parities. Since we defined $B_{global}$ so that it belongs to a conjugacy class distinct from those of $B_1,B_2,\dots,B_g$, the proof follows similarly. We will also use the fact that $C_0$ has minimum distance at least $h+2$ and so any $h+1$ columns of $H_0$ are linearly independent.
\end{proof}

\noindent \textbf{Case 2: } $q_0\ge g+1$ and $n_0=r+\ceil{\frac{h-1}{g}}$.\\
In this case, we don't have an extra (non-zero) conjugacy class to define $B_{global}$. Therefore, we will partition $B_{global}$ into $g$ parts and fold in the parts into the existing $g$ conjugacy classes. Note that we can always include the last column of $B_{global}$ as $[0,0,\dots,1]^T$. Therefore we only need to fold in $h-1$ columns of $B_{global}$ into existing $g$ conjugacy classes. Let $t=\ceil{\frac{h-1}{g}}$ and let $n_0 = r + t$.

Let $C_0$ be an $[n_0,n_0-c,d]_{\F_{q_0}}$ linear code with minimum distance $d\ge h+2$. Let $H_0$ be the parity check matrix of $C_0$ which is a $c \times n_0$ matrix over $\F_{q_0}$. By the hypothesis that the dual code of $C_0$ has a vector of weight exactly $r$, we can assume that the first row of $H_0$ has exactly $r$ ones and $t$ zeros (after scaling the columns if necessary). Partition $H_0$ as follows:
 \begin{equation}
 \label{eqn:H0_def_a1}
 H_0=\left[\begin{array}{c c c c| c c c c}
 	1 & 1 & \cdots & 1 & 0 & 0 &\cdots & 0\\
 	\hline
 	\beta_1 & \beta_2 & \cdots & \beta_r & \tbeta_1 & \tbeta_2 & \cdots & \tbeta_t
 \end{array}\right].
 \end{equation}
Define $$A_1,A_2,\dots,A_g=\begin{bmatrix}1&1&\cdots&1\end{bmatrix}.$$
Note that $\beta_1,\dots,\beta_r$ and $\tbeta_1,\dots,\tbeta_t$ can also be thought of as elements of $\F_{q_0^{c-1}}$ by fixing some basis of $\F_{q_0^{c-1}}$ as a vector space over $\F_{q_0}.$
Define for $1\le \ell \le g,$ $B_\ell$ as in (\ref{eq:choice-of-B-globaloutside}), $B_{global}^\ell$ as in (\ref{eq:choice-of-B-case2-globaloutside}) and $\widetilde{B}_{global}$ as in (\ref{eq:choice-of-tBglobal-case2}).
% \begin{equation}
% %\label{eq:choice-of-B}
% B_\ell=
% \begin{bmatrix}
% 	 	\beta_1& \beta_2 & \dots & \beta_{r}\\
% 	 	\gamma^{\ell}\beta_1^{{q_0}}& \gamma^{\ell}\beta_2^{{q_0}}& \cdots & \gamma^{\ell}\beta_{r}^{{q_0}}\\
% 	 	\gamma^{\ell(1+{q_0})}\beta_1^{q_0^2}& \gamma^{\ell(1+{q_0})}\beta_2^{q_0^2}& \cdots & \gamma^{\ell(1+{q_0})}\beta_{r}^{q_0^2}\\
% 	 	\vdots & \vdots &  &\vdots\\
% 	 	\gamma^{\ell(1+{q_0}+\dots+q_0^{h-2})}\beta_1^{q_0^{h-1}}& \gamma^{\ell(1+{q_0}+\dots+q_0^{h-2})}\beta_2^{q_0^{h-1}}& \cdots & \gamma^{\ell(1+{q_0}+\dots+q_0^{h-2})}\beta_{r}^{q_0^{h-1}}
% 	 \end{bmatrix},
% \end{equation}
% \begin{equation}
% %\label{eq:choice-of-B}
% B_{global}^\ell=
% \begin{bmatrix}
% 	 	\tbeta_1& \tbeta_2 & \dots & \tbeta_{t}\\
% 	 	\gamma^{\ell}\tbeta_1^{{q_0}}& \gamma^{\ell}\tbeta_2^{{q_0}}& \cdots & \gamma^{\ell}\tbeta_{t}^{{q_0}}\\
% 	 	\gamma^{\ell(1+{q_0})}\tbeta_1^{q_0^2}& \gamma^{\ell(1+{q_0})}\tbeta_2^{q_0^2}& \cdots & \gamma^{\ell(1+{q_0})}\tbeta_{t}^{q_0^2}\\
% 	 	\vdots & \vdots &  &\vdots\\
% 	 	\gamma^{\ell(1+{q_0}+\dots+q_0^{h-2})}\tbeta_1^{q_0^{h-1}}& \gamma^{\ell(1+{q_0}+\dots+q_0^{h-2})}\tbeta_2^{q_0^{h-1}}& \cdots & \gamma^{\ell(1+{q_0}+\dots+q_0^{h-2})}\tbeta_{t}^{q_0^{h-1}}
% 	 \end{bmatrix}.
% \end{equation}

% Define $$\widetilde{B}_{global}=
% \left[\begin{array}{c|c|c|c|c}
% 	B_{global}^1 & B_{global}^2 & \cdots & B_{global}^g & e_h
% \end{array}
% \right]
% $$
% where $e_h \in F_{q_0}^h$ is the coordinate vector $e_h=[0,0,\dots,0,1]^T$. 
Note that $\widetilde{B}_{global}$ is an $h\times (gt+1)$ matrix and $gt+1 \ge h$. Finally define $B_{global}$ to be an $h\times h$ matrix formed by arbitrary $h$ columns of $\widetilde{B}_{global}$.

\begin{claim}
	The above construction is an MR LRC over fields of size ${q}=q_0^{c-1}.$
\end{claim}
\begin{proof}[Proof Sketch]
The proof is very similar to that of Claim~\ref{claim:MRLRC_main}. The only difference is that some of the $h$ additional erasures can happen in the global parities. We will need to crucially use the fact the top right corner of the matrix $H_0$ in (\ref{eqn:H0_def_a1}) is zero. Therefore using some `$a$' ones to remove the rest of the ones in the upper half of $H_0$, does not affect the $\tbeta$ matrix since the top right corner is already forced to be zero.
\end{proof}

\section{Skew Polynomial Wronskian and Moore matrices}
\label{sec:wronskian-moore}
In this section, we will discuss generalizations of Wronskian and Moore matrices using skew polynomials. The non-singularity of special cases of these matrices has been instrumental in works on list decoding~\cite{GW-derivative-codes,GW13} and algebraic pseudorandomness such as constructions of rank condensers and subspace designs~\cite{FG15,GK-combinatorica,GXY-tams}. We will need the following simple lemmas.

\renewcommand{\L}{\mathbb{L}}
\begin{lemma}
	\label{lem:linear_independence_lowdegree}
	Let $\F(x)$ be the field of rational functions in $x$ and let $\L=\F(x^r)$ which is a subfield of $\F(x)$.\footnote{$\F(x^r)$ is the set of rational functions of the form $f(x^r)$ for $f\in \F(x)$ i.e. rational functions which only have terms whose powers are multiples of $r$.} Let $g_1,g_2,\dots,g_m \in \F[x]^{<r}$ be polynomials of degree strictly less than $r$. Then $g_1,g_2,\dots,g_m$ are $\L$-linearly independent iff they are $\F$-linearly independent.
\end{lemma}
\begin{proof}
	One direction is obvious since $\F$ is a subfield of $\L$. To prove the other direction, suppose $g_1,g_2,\dots,g_m$ are $\L$-linearly dependent, i.e., $\sum_i c_i(x^r) g_i(x) = 0$ for some $c_i\in \F(x).$ WLOG, by clearing denominators and common factors, we can assume that $c_i$ are also polynomials (i.e., $c_i\in \F[x]$) with no common factor. By comparing the coefficients of powers of $x$ between $0$ and $r-1$, we immediately get that $\sum_i c_{i}(0)g_i(x) = 0$. Note that all $c_{i}(0)$ cannot be zero simultaneously since then $x$ would be a common factor for all $c_i$. Therefore we get a non-trivial $\F$-linear dependency for $g_1,g_2,\dots,g_m$.
\end{proof}
	% One direction is obvious since $\F$ is a subfield of $\L$. To prove the other direction, suppose $g_1,g_2,\dots,g_m$ are $\L$-linearly dependent, i.e., $\sum_i c_i(x^r) g_i(x) = 0$ for some $c_i\in \F(x).$ WLOG, let's assume that $c_1\ne 0.$ We will prove that $g_1,g_2,\dots,g_m$ should also be $\F$-linearly dependent.

	% Let $c_i(x^r)= \sum_{j\in \Z} c_{ij} x^{rj}$ be the Laurent series expansion of $c_i$ where the index $j$ can take only finitely many negative values and $c_{ij} \in \F$. WLOG, by multiplying all the $c_i$ by powers of $x^r$, we can assume that $c_{10}\ne 0.$ By substituting these expansions for $c_i$ into $\sum_i c_i(x^r) g_i(x) = 0$ and comparing the coefficients of powers of $x$ between $0$ and $r-1$, we immediately get that $\sum_i c_{i0}g_i(x) = 0$. Therefore we get a non-trivial $\F$-linear dependency for $g_1,g_2,\dots,g_m$.

\begin{lemma}
\label{lem:phi_a_map}
	Let $\K[x;\sigma,\delta]$ be a skew polynomial ring. For $a\in \K$, define $\phi_a: \K \to \K$ as $\phi_a(y)=\sigma(y)a+\delta(y)$. Then
	\begin{enumerate}
		\item $\phi_a^i(y)=N_i(\conj{a}{y})y$ where $\phi_a^i$ is $\phi_a$ composed with itself $i$ times and
		\item $\phi_a$ is a linear map over the subfield $\K_a$.
	\end{enumerate}
\end{lemma}
\begin{proof}
	(1) This can be proved by induction, it is true for $i=1$. 
	\begin{align*}
	N_{i+1}(\conj{a}{y})y &= \sigma(N_i(\conj{a}{y}))\conj{a}{y}y + \delta(N_i(\conj{a}{y}))y\\
	 &= \sigma(N_i(\conj{a}{y}))(\sigma(y)a+\delta(y)) + \delta(N_i(\conj{a}{y}))y\\
	 & = \sigma(N_i(\conj{a}{y})y)a + \sigma(N_i(\conj{a}{y}))\delta(y) + \delta(N_i(\conj{a}{y}))y\\
	 & = \sigma(N_i(\conj{a}{y})y)a + \delta(N_i(\conj{a}{y})y)\\
	 & = \phi_a (N_i(\conj{a}{y})y) = \phi_a (\phi_a^i(y)) = \phi_a^{i+1}(y).
	 \end{align*}	
	 (2) $\K_a$-linearity follows since $\forall c\in \K_a$, $$\phi_a(yc)=N_1(\conj{a}{yc})yc = N_1(\conj{(\conj{a}{c})}{y})yc=N_1(\conj{a}{y})yc=\phi_a(y)c \ .  \qedhere $$
\end{proof}

Using Lemma~\ref{lem:phi_a_map}, one can linearize the evaluation of skew-polynomials on any conjugacy class. This gives a bijection between evaluation of skew-polynomials on a particular conjugacy class and \emph{linearized polynomials} which found several applications in coding theory and linear-algebraic pseudorandomness~\cite{mahdavifar2013algebraic,guruswami2018lossless,berlekamp2015algebraic}. In fact this is a ring isomorphism and the product operation denoted by $\otimes$ in \cite{mahdavifar2013algebraic} is equivalent to the product operation for skew polynomials in the appropriate skew polynomial ring.

\subsection{Wronskian matrix}
\label{sec:wronskian}
The theory of skew polynomials allows us to calculate rank of Wronskian matrices. Let $\K[x;\delta]$ be a skew-polynomial of derivation type i.e. $\sigma\equiv \Id$ is the identity map.
\begin{definition}[Wronskian]
		Let $c_1,\dots,c_n\in \K^*$. Define the Wronskian $$W_n(c_1,\dots,c_n)=
	\begin{bmatrix}
		c_1 & c_2 &\cdots &c_n\\
		\delta(c_1)& \delta(c_2) & \cdots & \delta(c_n)\\
		\delta^2(c_1)& \delta^2(c_2) & \cdots & \delta^2(c_n)\\
		\vdots & \vdots & & \vdots\\
		\delta^{n-1}(c_1)& \delta^{n-1}(c_2) & \cdots & \delta^{n-1}(c_n)\\
	\end{bmatrix}.$$
\end{definition}

\begin{corollary}
	\label{cor:wronskian}
	 $W_n(c_1,\dots,c_n)$ is full-rank iff $c_1,\dots,c_n$ are linearly independent over $\F=\K_0$, the centralizer of $0.$
\end{corollary}
\begin{proof}
	By Lemma~\ref{lem:phi_a_map}, $\delta^i(c)=N_i(\conj{0}{c})c.$ Thus the claim follows from Lemma~\ref{lem:rank_Vandermonde}.
\end{proof}

Note that when $\delta$ is the formal derivative of polynomials, the above is the usual Wronskian of polynomials. Applying the above corollary in this special case, we can relate the non-singularity of the Wronskian to the linear independence of the polynomials.
\begin{proposition}
	\label{prop:wronskian_linear_independence}
	Let $f_1,f_2,\dots,f_s\in \F[x]$ be polynomials of degree at most $d$. Suppose $\delta^j(f_i)$ is the $j^{th}$ derivative of $f_i$. Define 
	$$M= 
	\begin{bmatrix}
		f_1(x) & f_2(x) &\dots & f_s(x)\\
		\vdots  & \vdots&  & \vdots\\
		\delta^j(f_1)(x) &\delta^j(f_2)(x) & \dots & \delta^j(f_s)(x)\\
		\vdots & \vdots & & \vdots\\
		\delta^{s-1}(f_1)(x) &\delta^{s-1}(f_2)(x) & \dots & \delta^{s-1}(f_s)(x)\\
	\end{bmatrix}.
	$$ Then the following are true:
	\begin{enumerate}
		\item If $\charF(\F)=p$ then\footnote{$\charF(\F)$ is the characteristic of $\F$.}, $\det(M)\ne 0$ iff $f_1,f_2,\dots,f_2$ are linearly independent over $\F(x^p)$.
		\item If $\charF(\F)>d$ or $\charF(\F)=0$ then, $\det(M)\ne 0$ iff $f_1,f_2,\dots,f_s$ are linearly independent over $\F$.
	\end{enumerate}
\end{proposition}
\begin{proof}
	It is clear that if $f_1,f_2,\dots,f_d$ are linearly dependent over $\F$, then $\det M=0$. Now we will prove the converse.

	Consider the skew polynomial ring defined in Example~\ref{example:skewpolynomialring} where $\K=\F(x), \sigma\equiv \Id$ and $\delta(f)$ is the derivative of $f$. 
	By Corollary~\ref{cor:wronskian}, $\det(M)$ is zero iff $f_1,f_2,\dots,f_s$ are linearly independent over $\K_0,$ the centralizer of $0.$ We have
	$$\K_0=\{g: \conj{0}{g}=0\}\cup\{0\}=\{g: \delta(g)=0\}.$$
	If $\charF(\F)=0,$ then $\K_0=\F$ and we are done. If $\charF(\F)=p$ for some prime $p,$ then we claim below that $\K_0=\F(x^p)$, which finishes the proof using  Lemma~\ref{lem:linear_independence_lowdegree}.
        \end{proof}
	\begin{claim}
		If $\charF(\F)=p$, then $\K_0=\F(x^p).$
	\end{claim}
	\begin{proof}
		$\K_0=\{g\in \F(x): \delta(g)=0\}$. If $g\in \F[x]$, then it is easy to see that $\delta(g)=0$ iff $g\in \F[x^p].$ Now suppose $g$ is a rational function of the form $g=a/b$ where $a,b\in \F[x]$ do not have any common factors. By product rule, $\delta(g)=0\iff \delta(a)b = a \delta(b)$. Since $a,b$ do not have any common factors, this implies that $a$ divides $\delta(a)$ and $b$ divides $\delta(b).$ Since degree of $\delta(a)$ is smaller than $a$, this is not possible unless $\delta(a)=0$ and similarly we can conclude that $\delta(b)=0$. Therefore $a,b\in \F[x^p]$ and so $g\in \F(x^p).$
	\end{proof}
	
	%Note that $\F(x^p)$ is a subfield of $\F(x).$ By clearing denominators and removing common factors, we can conclude that $\det M = 0$ implies that there exists $c_1,\dots,c_s \in \F[x]$ (with no non-trivial common factor) such that $\sum_{i=1}^s c_i(x^p) f_i(x)=0.$ Now since the degree of each $f_i$ is less than $p$, comparing coefficients of terms up to degree $p-1,$ we get $\sum_{i=1}^s c_i(0) f_i(x)=0.$ Note that $c_1(0),c_2(0),\dots,c_s(0)$ cannot all simultaneously be zero because $c_1,c_2,\dots,c_s$ do not have any common factors. Therefore $f_1,f_2,\dots,f_s$ are linearly independent over $\F.$
%\end{proof}

Using the above, we can now deduce the following result which is the basis of list-size bound for list decoding univariate multiplicity codes~\cite{GW-derivative-codes} and the analysis of the associated subspace design constructed in \cite{GK-combinatorica}.

\begin{proposition}
	\label{prop:listdecoding_multiplicity}
	Let $\charF(\F)=p$. Let $\delta$ be the derivative operator on polynomials in $\F[x]$ and $\delta^i(\cdot)$ be the $i^{th}$ derivative of a polynomial. Let $Q(x,y_0,y_1,\dots,y_{s-1})=A(x)+\sum_{i=0}^{s-1} A_i(x)y_i$ where $A(x),A_i(x)\in \F[x]$ and not all $A_i$ are zero. The set of all $f\in \F[x]$ of degree less than $p$, such that 
	\begin{equation}
		\label{eqn:Q_multiplicity}
	Q(x,f(x),\delta(f)(x),\dots,\delta^{s-1}(f)(x))=0,	
	\end{equation}
	 form an $\F$-affine subspace of $\F[x]$ of dimension at most $s-1$.
\end{proposition}
\begin{proof}
	Equation (\ref{eqn:Q_multiplicity}) can be rewritten as $A+\sum_{i=0}^{s-1} A_i \delta^{i}(f)=0.$ Suppose that the set of solutions to this equation in $\F[x]^{<p}$ form an $\F$-affine subspace of $\F[x]$ of dimension at least $s$. Then there exist solutions $f_0,f_1,\dots,f_s \in \F[x]^{<p}$ where $f_1-f_0,\dots,f_s-f_0$ are $\F$-linearly independent. Let $g_i = f_i-f_0.$ Then for $j\in [s]$ we have, $\sum_{i=0}^{s-1}A_i \delta^i(g_j)=0.$ Therefore the determinant of the matrix $[\delta^i(g_j)]_{ij}$ is zero. Therefore by Proposition~\ref{prop:wronskian_linear_independence}, $g_1,g_2,\dots,g_s$ should be $\F$-linearly dependent, which is a contradiction.
\end{proof}

We also remark that solving equation~(\ref{eqn:Q_multiplicity}) when $A=0$ is equivalent to finding roots of a skew polynomial of degree $s-1$ in a conjugacy class. This also intuitively explains why the set of solutions is an affine subspace of dimension at most $s-1$.
Consider the skew polynomial ring $\K[t;\delta]$ of derivation type where $\K=\F(x)$, $\sigma\equiv \Id$ and $\delta$ is the derivative operator. Then by Lemma~\ref{lem:phi_a_map}, $N_i(\conj{0}{f})f=\delta^{i}(f).$ Therefore the Equation~(\ref{eqn:Q_multiplicity}), when $A=0$, can be rewritten as:
$$\sum_{i=0}^{s-1} A_i \delta^{i}(f)=0 \iff \sum_{i=0}^{s-1} A_i N_{i}(\conj{0}{f})f=0.$$
Define $G(t)\in \K[t;\delta]$ as $G(t)=\sum_{i=0}^{s-1} A_i t^i$ which is a skew polynomial of degree at most $s-1.$ Then $G(\conj{0}{f}) f=\sum_{i=0}^{s-1}A_i N_i(\conj{0}{f})f.$ Therefore the solutions of (\ref{eqn:Q_multiplicity}) when $A=0$ are precisely $\{0\}\cup \{f: G(\conj{0}{f}) = 0\}.$

\subsection{Moore matrix}
\label{sec:moore}
The theory of skew polynomials also allows us to calculate the rank of Moore matrices.
Let $\K[t;\sigma]$ be a skew polynomial ring of endomorphism type i.e. $\delta\equiv 0$. This is completely analogous to Wronskian matrices (Section~\ref{sec:wronskian}) once we use the skew polynomial framework.

\begin{definition}[Moore matrix]
		 Let $c_1,\dots,c_n\in \K^*$. Define the Moore matrix $$M_n(c_1,\dots,c_n)=
	\begin{bmatrix}
		c_1 & c_2 &\cdots &c_n\\
		\sigma(c_1)& \sigma(c_2) & \cdots & \sigma(c_n)\\
		\sigma^2(c_1)& \sigma^2(c_2) & \cdots & \sigma^2(c_n)\\
		\vdots & \vdots & & \vdots\\
		\sigma^{n-1}(c_1)& \sigma^{n-1}(c_2) & \cdots & \sigma^{n-1}(c_n)\\
	\end{bmatrix}.
	$$
\end{definition}
\begin{corollary}
	\label{cor:moore_determinant}
 $M_n(c_1,\dots,c_n)$ is full-rank iff $c_1,\dots,c_n$ are linearly independent over $\F=\K_1,$ the centralizer of $1.$
\end{corollary}
\begin{proof}
	By Lemma~\ref{lem:phi_a_map}, $\sigma^i(c)=N_i(\conj{1}{c})c$. Thus the claim follows from Lemma~\ref{lem:rank_Vandermonde}.
\end{proof}

We now apply the above to the case when $\K = \F_q(x)$ and $\sigma$ is the automorphism which maps $f(x)\in \F_q(x)$ to $f(\gamma x)$ for a generator $\gamma$ of $\F_q^*$. In this case, the Moore matrix was called the folded Wronskian in \cite{GK-combinatorica}. Analogous Moore matrices for function fields were studied in \cite{GXY-tams}.

\begin{proposition}
	\label{prop:folded_wronskian_linear_independence}
	Let $f_1,f_2,\dots,f_s\in \F_q[x]$ be polynomials of degree at most $d$. Let $\gamma$ be generator for $\F_q^*.$  Define 
	$$M= 
	\begin{bmatrix}
		f_1(x) & f_2(x) &\dots & f_s(x)\\
		\vdots  & \vdots&  & \vdots\\
		f_1(\gamma^j x) &f_2(\gamma^j x) & \dots & f_s(\gamma^j x)\\
		\vdots & \vdots & & \vdots\\
		f_1(\gamma^{s-1} x) &f_2(\gamma^{s-1} x) & \dots & f_s(\gamma^{s-1} x)\\	
	\end{bmatrix}.
	$$
	Then the following are true:
	\begin{enumerate}
		\item $\det(M)\ne 0$ iff $f_1,f_2,\dots,f_2$ are linearly independent over $\F_q(x^{q-1})$.
		\item If $q-1>d$ then, $\det(M)\ne 0$ iff $f_1,f_2,\dots,f_s$ are linearly independent over $\F_q$.
	\end{enumerate}
\end{proposition}
\begin{proof}
	It is clear that if $f_1,f_2,\dots,f_d$ are linearly dependent over $\F_q$, then $\det M=0$. Now we will prove the converse.

	Consider the skew polynomial ring defined in Example~\ref{example:skewpolynomialring} where $\K=\F_q(x), \sigma(g(x))=g(\gamma x)$ and $\delta\equiv 0$. 
	By Corollary~\ref{cor:moore_determinant}, $\det(M)$ is zero iff $f_1,f_2,\dots,f_s$ are linearly independent over $\K_1,$ the centralizer of $1.$ We have
	$$\K_1=\{g: \conj{1}{g}=1\}\cup\{0\}=\{g: g(\gamma x)=g(x)\}.$$
        We now claim that $\K_1 = \F_q(x^{q-1})$ and the rest follows from Lemma~\ref{lem:linear_independence_lowdegree}.
        \end{proof}

	\begin{claim}
		$\K_1=\F_q(x^{q-1}).$
	\end{claim}
	\begin{proof}
		$\K_1=\{g\in \F_q(x): g(\gamma x)=g(x)\}$. If $g\in \F_q[x]$, then it is easy to see that $g(\gamma x)=g(x)$ iff $g\in \F_q[x^{q-1}].$ Now suppose $g$ is a rational function of the form $g=a/b$ where $a,b\in \F_q[x]$ do not have any common factors and we can assume that the constant term of $a$ or $b$ is $1$. $g(\gamma x) = g(x) \iff a(\gamma x)b(x) = a(x) b(\gamma x)$. Since $a,b$ do not have any common factors, this implies that $a$ divides $a(\gamma x)$ and $b$ divides $b(\gamma x).$ Since degree of $a(\gamma x)$ is the same as that of $a(x)$ and the degree of $b(\gamma x)$ is the same as that of $b(x)$, this implies that $a(\gamma x)=\lambda a(x)$ and $b(\gamma x) = \lambda b(x)$ for some $\lambda \in \F_q$. Since we assumed that $a$ or $b$ has constant term 1, we can conclude that $\lambda=1$. Therefore $a,b\in \F_q[x^{q-1}]$ and so $g\in \F_q(x^{q-1}).$
	\end{proof}
	% By clearing denominators and removing common factors, we can conclude that $\det M = 0$ implies that there exists $c_1,\dots,c_s \in \F[x]$ (with no non-trivial common factor) such that $\sum_{i=1}^s c_i(x^{q-1}) f_i(x)=0.$ Now since the degree of each $f_i$ is less than $q-1$, comparing coefficients of terms up to degree $q-2,$ we get $\sum_{i=1}^s c_i(0) f_i(x)=0.$ Note that $c_1(0),c_2(0),\dots,c_s(0)$ cannot all simultaneously be zero because $c_1,c_2,\dots,c_s$ do not have any common factors. Therefore $f_1,f_2,\dots,f_s$ are linearly independent over $\F.$
%\end{proof}

Using the above, we can now deduce the following result which is the basis of list-size bound for list decoding folded Reed-Solomon codes~\cite{Gur-ccc11,GW13} and the analysis of the subspace design constructed using folded Reed-Solomon codes~\cite{GK-combinatorica}.

\begin{lemma}
	\label{lem:listdecoding_foldedRS}
	Let $\gamma$ be a generator for $\F_q^*$. Let $Q(x,y_0,y_1,\dots,y_{s-1})=A(x)+\sum_{i=0}^{s-1} A_i(x)y_i$ where $A(x),A_i(x)\in \F_q[x]$ and not all $A_i$ are zero. The set of all $f\in \F_q[x]$ of degree less than $q-1$, such that 
	\begin{equation}
		\label{eqn:Q_foldedRS}
		Q(x,f(x),f(\gamma x),\dots,f(\gamma^{s-1} x))=0,
	\end{equation}
	 form an $\F_q$-affine subspace of $\F_q[x]$ of dimension at most $s-1$.
\end{lemma}
\begin{proof}
		Equation (\ref{eqn:Q_foldedRS}) can be rewritten as $A+\sum_{i=0}^{s-1} A_i f(\gamma^i x)=0.$ Suppose that the set of solutions to this equation in $\F_q[x]^{<q-1}$ form an $\F_q$-affine subspace of $\F_q[x]$ of dimension at least $s$. Then there exist solutions $f_0,f_1,\dots,f_s \in \F_q[x]^{<q-1}$ where $f_1-f_0,\dots,f_s-f_0$ are $\F_q$-linearly independent. Let $g_i = f_i-f_0.$ Then for $j\in [s]$ we have, $\sum_{i=0}^{s-1}A_i g_j(\gamma^i x)=0.$ Therefore the determinant of the matrix $[g_j(\gamma^i x)]_{ij}$ is zero. Therefore by Proposition~\ref{prop:folded_wronskian_linear_independence}, $g_1,g_2,\dots,g_s$ should be $\F_q$-linearly dependent, which is a contradiction.
\end{proof}
Just as we did in Section~\ref{sec:wronskian}, we remark that solving Equation~(\ref{eqn:Q_foldedRS}), when $A=0$, is equivalent to finding roots of the degree $s-1$ skew polynomial $G(t)=\sum_{i=0}^{s-1}A_i t^i$ in the conjugacy class of $1$, where the underlying skew polynomial ring is $\K[t;\sigma]$ where $\K=\F(x)$ and $\sigma(f(x))=f(\gamma x).$

\section{Maximum sum rank distance codes}
\label{sec:MSRD}
\newcommand{\sumrk}{\mathrm{sum\text{-}rank}}
In this section, we will present a construction of Maximum Sum-Rank Distance (MSRD) codes due to~\cite{Martinez18} using the skew polynomial framework. We will first define sum-rank distance codes.
	
	Fix some basis $\cB$ for $\F_{q^m}$ as vector space over $\F_q.$ Given $z=(z_1,z_2,\dots,z_r)\in \F_{q^m}^r$, we can think of $z$ as an $m\times r$ matrix with entries in $\F_q$ by expressing each coordinate $z_i$ as a $\F_q^m$ vector using basis $\cB$; define $\rk_{\F_q}(z)$ to be the $\F_q$-rank of that matrix.
	Let $\cP=A_1\sqcup A_2 \sqcup \dots \sqcup A_s$ be a partition of $[n]$ into $s$ parts. Given $x\in \F_{q^m}^n$, let $x=(x_1,x_2,\dots,x_s)$ be the partition of of $x$ according to $\cP$ where $x_i \in \F_{q^m}^{A_i}$. Define $\sumrk_\cP(x)=\sum_{i=1}^s \rk_{\F_q}(x_i).$
\begin{definition}[sum-rank distance]
	Fix some partition $\cP=A_1\sqcup A_2 \sqcup \dots \sqcup A_s$ of $[n]$ into $s$ parts. An $\F_{q^m}$-linear subspace $C$ of $\F_{q^m}^n$ is said to have sum-rank distance $d$ (w.r.t. partition $\cP$) if every non-zero codeword $c\in C$, $\sumrk_\cP(c) \ge d.$
\end{definition}
Note that the sum-rank distance generalizes both Hamming metric (by choosing $\cP=\{1\}\sqcup \{2\} \sqcup \dots \sqcup \{n\}$) and rank metric (by choosing $\cP=[n]$). Moreover for any partition $\cP$ and any $x\in \F_{q^m}^n$, $\sumrk_\cP(x)$ is most the Hamming weight of $x$ (as rank is upper bounded by the number of non-zero columns). Therefore by the Singleton bound, any $k$-dimensional code of $\F_{q^m}^n$, can have sum-rank distance at most $n-k+1$. A code achieving this bound is called an MSRD code. Therefore MSRD codes generalize both MDS codes and Gabidulin codes. Sum-rank distance was introduced by~\cite{NU10} for applications in network coding. We will now present the construction of MSRD codes.

%\vnote{Give a quick definition of sum rank metric codes? Perhaps restrict to $\F_q$-linear codes?}

\begin{theorem}[Construction of maximum sum rank distance codes~\cite{Martinez18}]
\label{thm:max_sum_rank_distance_codes}
Let $\gamma$ be a generator for $\F_{q^m}$ and let $\beta_1,\dots,\beta_m \in \F_{q^m}$ be linearly independent over $\F_q.$ Let $n=(q-1)m$. For $k\le n,$  define a $k\times n$ matrix $M=[M_0|M_1|\dots|M_{q-2}]$ where 
$$
M_\ell=
\begin{bmatrix}
	 	\beta_1& \beta_2 & \dots & \beta_m\\
	 	\gamma^{\ell}\beta_1^{q}& \gamma^{\ell}\beta_2^{q}& \cdots & \gamma^{\ell}\beta_m^{q}\\
	 	\gamma^{\ell(1+q)}\beta_1^{q^2}& \gamma^{\ell(1+q)}\beta_2^{q^2}& \cdots & \gamma^{\ell(1+q)}\beta_m^{q^2}\\
	 	\vdots & \vdots &  &\vdots\\
	 	\gamma^{\ell(1+q+\dots+q^{k-2})}\beta_1^{q^{k-1}}& \gamma^{\ell(1+q+\dots+q^{k-2})}\beta_2^{q^{k-1}}& \cdots & \gamma^{\ell(1+q+\dots+q^{k-2})}\beta_m^{q^{k-1}}
	 \end{bmatrix}.
$$
Then $M$ is the generator matrix of a maximum sum rank distance code, i.e., for every non-zero vector $\lambda\in \F_{q^m}^k,$ $\sum_{\ell=0}^{q-2}\rk_{\F_q}(\lambda^T M_\ell)\ge n-k+1.$\footnote{Here we are interpreting a row vector $c\in \F_{q^m}^r$ as an $m\times r$ matrix over $\F_q$. $\rk_{\F_q}(c)$ is the $\F_q$-rank of this matrix. We will also use $\ker_{\F_q}(c)$ in the proof to denote the kernel of the matrix.}
\end{theorem}
\begin{proof}
	Suppose $\lambda \in \F_{q^m}^k$ is a non-zero vector such that $\sum_{\ell=0}^{q-2}\rk_{\F_q}(\lambda^T M_\ell)\le n-k.$ This is equivalent to $\sum_{\ell=0}^{q-2}\dim_{\F_q}(\ker_{\F_q}(\lambda^T M_\ell))\ge k.$

	Let $\K=\F_{q^m}$, $\sigma(a)=a^q$ and $\delta\equiv 0$. See Example~\ref{example:Frobenious} for the conjugation relation and conjugacy classes in this case. Define $f(t)=\sum_{i=0}^{k-1} \lambda_i t^i$ which is a non-zero skew polynomial of degree at most $k-1$ in $\F_{q^m}[t;\sigma]$. We will find many roots for $f$ which would violate Theorem~\ref{thm:fundamentalthm_roots_skewpolynomials} to get a contradiction.

	Fix some $\ell\in \set{0,1,\dots,q-2}$. Suppose $\dim_{\F_q}(\ker_{\F_q}(\lambda^T M_\ell))=d_\ell$. Let $\mu_1,\dots,\mu_{d_\ell} \in \F_q^m$ be a basis for the kernel. Let $\beta=(\beta_1,\beta_2,\dots,\beta_m)\in \F_{q^m}^m$. Now $\lambda^T M_\ell \mu_i=0$ implies that $\beta^T\mu_i$ is root of $f$. Moreover the $d_\ell$ roots $\beta^T\mu_1,\dots, \beta^T\mu_{d_\ell}\in \F_{q^m}$ are linearly independent over $\F_q$ since $\rk_{\F_q}(\beta)=m$.

	Thus we get $\sum_{\ell=0}^{q-2} d_\ell \ge k$ roots for $f$. And the roots in each conjugacy class are linearly independent over $\F_q$ (which is the centralizer). Therefore by Theorem~\ref{thm:fundamentalthm_roots_skewpolynomials}, we get a contradiction.
\end{proof}

It is easy to see that the above construction can be easily modified to work for any partition $\cP$ of $[n]$ into at most $(q-1)$ parts, where each part has size at most $m.$
In~\cite{MartinezK19_NetworkCoding}, an efficient decoding algorithm for these codes is given.

\end{document}